\documentclass[submission,copyright,creativecommons]{eptcs}
\providecommand{\event}{} 
\usepackage{breakurl}             
\usepackage{underscore}           
\usepackage{amsmath}
\usepackage[dvipdfmx]{graphicx}
\usepackage[x11names]{xcolor}
\usepackage{wrapfig}
\usepackage{tikz-cd}
\usepackage{comment}
\usepackage{tikz}
\usepackage[labelformat=simple]{subcaption}

\usepackage{xspace}
\usepackage{cmll}
\usepackage{proof}
\usepackage{cleveref}
\usepackage{stmaryrd}
\usepackage{scalerel}
\usepackage[noadjust]{cite}
\usepackage{mathbbol}
\usepackage{amsmath}
\usepackage{amsthm} 
\usepackage{amssymb}

\DeclareSymbolFontAlphabet{\mathbbm}{bbold}
\usepackage{accsupp} 
\newcommand*{\llbrace}{%
  \BeginAccSupp{method=hex,unicode,ActualText=2983}%
    \textnormal{\usefont{OMS}{lmr}{m}{n}\char102}%
    \mathchoice{\mkern-4.05mu}{\mkern-4.05mu}{\mkern-4.3mu}{\mkern-4.8mu}%
    \textnormal{\usefont{OMS}{lmr}{m}{n}\char106}%
  \EndAccSupp{}%
}
\newcommand*{\rrbrace}{%
  \BeginAccSupp{method=hex,unicode,ActualText=2984}%
    \textnormal{\usefont{OMS}{lmr}{m}{n}\char106}%
    \mathchoice{\mkern-4.05mu}{\mkern-4.05mu}{\mkern-4.3mu}{\mkern-4.8mu}%
    \textnormal{\usefont{OMS}{lmr}{m}{n}\char103}%
  \EndAccSupp{}%
}

\newsavebox{\lXbrace}
\savebox{\lXbrace}{$\llbrace$}
\newsavebox{\rXbrace}
\savebox{\rXbrace}{$\rrbrace$}
\def\lxbrace{\scalerel*{\usebox{\lXbrace}}{\llbrace}}
\def\rxbrace{\scalerel*{\usebox{\rXbrace}}{\rrbrace}}

\usetikzlibrary{intersections, calc, arrows.meta}
\allowdisplaybreaks[4]

\newif\ifdraft\draftfalse

\newif\iffull\fullfalse 

\ifdraft
\newcommand\asd[1]{{\footnotesize \color[RGB]{105,10,130}[#1 -asd]}}
\newcommand\kzk[1]{{\footnotesize \color{blue}[#1 -kzk]}}
\newcommand\clovis[1]{{\footnotesize \color[RGB]{0,200,30}[#1 -clovis]}}
\newcommand\ichiro[1]{{\footnotesize \color[RGB]{105,10,10}[#1 -ichiro]}}

\newcommand\modi[4]{{\color{gray}[{\color{Brown4}#2}\ensuremath{\Rightarrow}{\color{red}#3}|{\color{green}#4{}
}-#1]}}
\newcommand\del[2]{\modi{#1}{#2}{}{}}
\else 
\newcommand\asd[1]{}
\newcommand\kzk[1]{}
\newcommand\clovis[1]{}
\newcommand\ichiro[1]{}

\newcommand\modi[4]{#3}
\newcommand\del[2]{}
\fi

\renewcommand{\cref}[1]{\Cref{#1}}
\creflabelformat{enumi}{#2(#1)#3}
\crefname{theorem}{Theorem}{Theorems}
\newcommand{\dt}[1][]{\textbf{(#1)}} 
\newcommand{\dtb}[1]{\textbf{(#1)}} 

\newcommand{\uval}{\star} 
\newcommand{\C}{\mathbb{C}}
\newcommand{\D}{\mathbb{D}}
\newcommand{\F}{\mathbb{F}} 
\newcommand{\U}{\mathbb{U}} 
\newcommand{\Ft}{\F_\textnormal{tr}} 
\newcommand{\Fs}{\F_\textnormal{smc}} 
\newcommand{\Usig}{\U_\textnormal{sig}} 
\newcommand{\Fsig}{\F_\textnormal{sig}} 
\newcommand{\sig}{\mathrm{sig}}

\newcommand{\set}[1]{\left\{{#1}\right\}}
\newcommand{\setcomp}[2]{\left\{{#1}\,\middle|\,{#2}\right\}}
\newcommand{\us}[1]{{|}#1{|}} 
\newcommand{\Nat}{{\mathbb N}}
\newcommand{\Nato}{\Nat_{\geq 1}}
\newcommand{\Natm}[1]{\Nat_{#1}}

\newcommand{\nset}[1]{[#1]}
\newcommand{\vect}[2]{\mathbf{#1}_{#2}}
\newcommand{\ol}[1]{\overline{#1}}
\newcommand{\defeq}{:=}

\newcommand{\op}{\mathit{op}}
\renewcommand{\parallel}{\oplus}
\newcommand{\streve}[1]{\mathrm{Str}_{\eve}(#1)}
\newcommand{\stradam}[1]{\mathrm{Str}_{\adam}(#1)}

\newcommand{\lplus}[2]{#2^{\downarrow #1}}
\newcommand{\lose}{\mathbf{lose}}
\newcommand{\g}[1]{\mathcal{#1}}
\newcommand{\Pe}{\mathrm{Play}_{\eve}}
\newcommand{\Pa}{\mathrm{Play}_{\adam}}
\newcommand{\play}[3]{\mathrm{play}^{#2,#3}_{#1}}

\newcommand{\dr}{\mathbbm{r}}
\newcommand{\dl}{\mathbbm{l}}
\newcommand{\dropg}{\mathbbm{r}}

\newcommand{\eve}{\exists}
\newcommand{\adam}{\forall}

\newcommand{\bfn}[2][]{\overline{#2}^{#1}}

\newcommand{\nfl}[2][]{#2^{#1}_{\dl}}
\newcommand{\nfr}[2][]{#2^{#1}_{\dr}}
\newcommand{\dual}[1]{#1^{\bot}}
\newcommand{\coset}[1]{[#1]_{\sim}}
\newcommand{\ID}{\mathcal{I}}
\newcommand{\initf}[1]{!^{\emptyset}_{#1}}
\newcommand{\leaves}[1]{\mathrm{leaves}(#1)}
\newcommand{\leavesf}[1]{\mathrm{leaves}^{\fix}(#1)}
\newcommand{\comp}[4]{\mathrm{SRT}^{\fix}(#1, #2, #3, #4)} 
\newcommand{\tcomp}[5]{\mathrm{SRT}(#1, #2, #3, #4, #5)}

\newcommand{\ccsp}{\Sigma_{\parity}}
\newcommand{\emptyword}{\varepsilon}
\newcommand{\nd}[3][]{\mathrm{n}^{#1}_{#2,#3}}
\newcommand{\drf}[1]{\overrightarrow{#1}}
\newcommand{\dlf}[1]{\overleftarrow{#1}}
\newcommand{\cplus}{\oplus} 

\newcommand{\hcomp}{*} 
\newcommand{\sym}{\sigma} 
\newcommand{\tp}[1]{#1^{\dagger}} 
\newcommand{\tr}{\mathrm{tr}}
\newcommand{\sm}{\mathrm{smc}}
\newcommand{\trace}[4]{\tr^{#4}_{#1;#2,#3}}
\newcommand{\tropg}[3]{\trace{#1}{#2}{#3}{}
}
\newcommand{\trcopg}[3]{\trace{#1}{#2}{#3}{\tcopg}}
\newcommand{\trsopg}[3]{\trace{#1}{#2}{#3}{\GM}}
\newcommand{\relto}{-\mspace{-13mu}+\mspace{-18mu}\to}
\newcommand{\id}{\mathrm{id}}
\newcommand{\Id}{\mathrm{Id}} 
\newcommand{\inj}{\mathrm{in}}

\newcommand{\Int}{\mathrm{Int}}
\newcommand{\Prop}{\mathbf{Prop}}
\newcommand{\Cprop}[1]{#1\text{-}\Prop}

\newcommand{\Csig}[1]{#1\text{-}\mathbf{Sig}}
\newcommand{\Cptcc}{\mathbf{CpCC}}
\newcommand{\CompCpt}[3][]{(#2 \circ #3)_{#1}}
\newcommand{\Comp}[2]{#1 \circ #2}
\newcommand{\csmt}[1]{\(#1\)-SMT\xspace}

\newcommand{\munit}{\mathbf{I}}
\newcommand{\assoc}{\mathbf{a}}
\newcommand{\lunit}{\mathbf{l}}
\newcommand{\runit}{\mathbf{r}}
\newcommand{\parity}{M}
\newcommand{\dunit}{d} 
\newcommand{\dcounit}{e} 

\newcommand{\Set}{\mathbf{Set}}
\newcommand{\Cat}{\mathbf{Cat}}
\newcommand{\FinPreord}{\mathbf{FinPreord}}
\newcommand{\scottl}{\mathbf{ScottL}}
\newcommand{\fscottl}{\mathbf{FinScottL}}
\newcommand{\colmm}{\square_M}
\newcommand{\GM}{\mathbf{GM}}
\newcommand\Pfin{P} 
\newcommand\SPfin{\mathcal{P}} 
\newcommand{\fix}{\mathrm{fix}}
\newcommand{\cpt}{\mathrm{cpt}} 
\newcommand{\iso}{g}
\newcommand{\adduc}[1]{#1^{\mathrm{de}}}
\newcommand{\adduce}{\adduc{E}_M}
\newcommand{\adducep}[1]{\adduc{E}_{#1}}
\newcommand{\adduct}[1]{\adduc{T}_{#1}}
\newcommand{\opgsym}{\mathrm{opg}}
\newcommand{\opgsig}{\Sigma^{\opgsym}_{M}}
\newcommand{\opgequ}{E^{\opgsym}_{M}}
\newcommand{\denotlr}[2]{%
\def\core{{\strut}#2{\strut}}%
\scalerel[#1]{\lxbrace}{\core}\scalerel*[#1]{\rxbrace}{\core}}
\newcommand{\deplay}[3][1.2ex]{\denotlr{#1}{#2}_{#3}}
\newcommand{\destr}[2][1.2ex]{\denotlr{#1}{#2}}
\newcommand{\subopg}[1]{#1}

\newcommand{\opg}{\mathbf{OPG}_M}
\newcommand{\topg}{\opg^{\dropg}}
\newcommand{\copg}{\F(\opgsig,\opgequ)}
\newcommand{\ints}{\mathrm{int}}
\newcommand{\inls}{\mathrm{inl}}
\newcommand{\intsig}[1]{#1^{\ints}}
\newcommand{\intsigp}[2]{#1^{\ints_{#2}}}
\newcommand{\tcopg}{\Ft(\intsig{\ccsp}) }
\newcommand{\tsopg}{\tsopgnoop^{\op}}
\newcommand{\tsopgnoop}{\fscottl_{!_M}}
\newcommand{\sopg}{\Int(\tsopg)}
\newcommand{\cto}{\mathcal{R}_M}
\newcommand{\tcto}{\mathcal{R}^{\dr}_M}
\newcommand{\tots}{\mathcal{W}^{\dr}_M}
\newcommand{\ots}{\mathcal{W}_M}
\newcommand{\cts}[1]{\left\llbracket{#1}\right\rrbracket_M}

\newcommand{\TrSMC}{\mathbf{TrSMC}}
\newcommand{\SMC}{\mathbf{SMC}}
\newcommand{\TrSMCg}{\TrSMC_g}
\newcommand{\CanFunc}[1]{\eta^{\Int}_{#1}}

\newlength{\hlength}
\newlength{\vlength}
\newlength{\nodesize}
\newcommand{\leftpos}[1]{#1}
\newcommand{\rightpos}[1]{#1'}

\theoremstyle{plain}
\newtheorem{theorem}{Theorem}[section]
\newtheorem*{theorem*}{Theorem}
\newtheorem{proposition}[theorem]{Proposition}
\newtheorem{lemma}[theorem]{Lemma}
\newtheorem{corollary}[theorem]{Corollary}

\theoremstyle{definition}
\newtheorem{definition}[theorem]{Definition}
\newtheorem{example}[theorem]{Example}
\newtheorem{remark}[theorem]{Remark}

\title{A Compositional Approach to Parity Games\footnote{The full version of the present paper is available at \url{https://group-mmm.org/\~kazuki/}.}\ \footnote{The authors are supported by ERATO HASUO
Metamathematics for Systems Design Project (No.\ JPMJER1603), JST Moonshot R\&D (JPMJMS2033-04), and JSPS KAKENHI Grant No.\ 18K11156.}}
\author{Kazuki Watanabe
\institute{The Graduate University for Advanced Studies (SOKENDAI)\\
Hayama, Japan}
\institute{National Institute of Informatics\\
Tokyo, Japan}
\email{kazukiwatanabe@nii.ac.jp}
\and
Clovis Eberhart
\institute{National Institute of Informatics\\
Tokyo, Japan}
\institute{Japanese-French Laboratory of Informatics\\
IRL 3527, CNRS, Tokyo, Japan}
\email{eberhart@nii.ac.jp}
\and
Kazuyuki Asada
\institute{Tohoku University\\
Sendai, Japan}
\email{asada@riec.tohoku.ac.jp}
\and
Ichiro Hasuo
\institute{The Graduate University for Advanced Studies (SOKENDAI)\\
Hayama, Japan}
\institute{National Institute of Informatics\\
Tokyo, Japan}
\institute{Ritsumeikan University\\
Kusatsu, Japan}
\email{i.hasuo@acm.org}
}
\def\titlerunning{A Compositional Approach to Parity Games}
\def\authorrunning{K. Watanabe, E. Clovis, K. Asada \& I. Hasuo}
\begin{document}
\maketitle

\begin{abstract}
  In this paper, we introduce open parity games, 
  which is a compositional approach to parity games.
  This is achieved by adding open ends
  to the usual notion of parity games.
  We introduce the category of open parity games, which is
  defined using standard definitions for graph games.
  We also define a graphical language for open parity games as a prop,
  which have recently been used in many applications as graphical
  languages.
  We introduce a suitable semantic category inspired by the work by
  Grellois and Melli{\`e}s on the semantics of higher-order model checking.
  Computing the set of winning positions in open parity games yields a
  functor to the semantic category.
  Finally, by interpreting the graphical language in the semantic
  category, we show that this computation can be carried out
  compositionally.
\end{abstract}

\section{Introduction}
\ichiro{The first paragraphs can be further extended. Paragraph 1: parity games. Paragraph 2: compositionality in model checking and parity games (earlier ones require finding an intermediate spec; Grellois \& Mellies is different, takes a categorical step). Paragraph 3: props. Paragraph 4: hinting what we do using the context that we discussed (in other words, why these existing results can give an answer to our question). Paragraph 5: What we do. }

\emph{Parity game} is a major tool in theoretical computer science.  Many formal verification problems such as model checking, satisfiability, etc.---can be reduced to solving parity games~\cite{wilke2001alternating}, where alternation of least and greatest fixed point operators in a specification is modeled by the parity winning condition. Efficient solutions of parity games, therefore, benefit many problems; recent algorithmic works include~\cite{CzerwinskiDFJLP19}.

In this paper, we are interested in \emph{compositionality} in formal verification in general, and in parity games in particular. It means that the property of a big system can be deduced from those of its constituent parts. One benefit is \emph{efficiency}: compositionality can yield an efficient divide-and-conquer algorithm. Another is \emph{maintainability}: compositional verification  explicates an assumption that each subsystem must satisfy for the safety of the whole system; a subsystem can then be replaced  freely as long as the local assumption is satisfied.

Compositional methods in model checking have been pursued in the literature, such as~\cite{ClarkeLM89,KwiatkowskaNPQ13}. Many of those methods require a user to provide interfaces between subsystems, either as systems~\cite{ClarkeLM89} or as specifications~\cite{KwiatkowskaNPQ13}. The role of compositionality is stressed in \emph{higher-order model checking (HOMC)}~\cite{grellois2015finitary,DBLP:conf/csl/TsukadaO14}, too, where intermediate results are combined along typing rules. 

In this paper, influenced by the semantical constructs from~\cite{grellois2015finitary}, we introduce a categorical framework in which parity games are both presented and  solved in a compositional manner. The presentation is by a \emph{prop}~\cite{maclane1965categorical} (products and permutations category), a categorical notion of ``monoidal'' algebraic structure. This categorical presentation enables us to formulate compositionality as the preservation of suitable structures of certain functors. It also enables us to exploit general categorical structures (traced, compact closed, etc.) and properties (such as freeness). 
 The use of props as graphical languages for various mathematical structures has been actively pursued recently (such as signal flow diagrams, matrices, and network games)~\cite{bonchi2017,bonchi2019graphical,lavore2021}; the current work adds a new item to the list, namely parity games.

\vspace*{.3em}
\noindent\textbf{Contribution.}\;
The outline of our paper is Fig.~\ref{fig:diagram}.
 We extend parity games with so-called \emph{open ends} so that we can compose them. The resulting notion \emph{(open parity game)} is organized in a compact closed category denoted by $\opg$. 
As a graphical language for open parity games, we use the prop $\copg$ freely generated by a suitable monoidal (algebraic) theory \((\opgsig,\opgequ)\). The other category $\sopg$ in Fig.~\ref{fig:diagram} originates from~\cite{grellois2015finitary}---it is our \emph{semantic category} that tells
which player is winning for (closed) parity games; for open parity games, it
provides intermediate results of a suitable granularity to
decide winners later.

 Our main theorem (Thm.~\ref{thm:triangleOfFunctors}) is the commutativity of Fig.~\ref{fig:diagram}; it says that the semantics of parity games $\ots$---defined as usual in terms of plays, strategies, and the parity acceptance condition---can be computed compositionally by a compact closed functor $\cts{-}$.
 The last compositional computation is illustrated in Ex.~\ref{ex:extExample}. 
After all, in the framework in Fig.~\ref{fig:diagram}, one writes down a parity game as a composition of smaller ones,  in the graphical language of the prop $\copg$; when it comes to solving games, the winning positions for 
larger games are computed from those of the smaller ones, using that $\cts{-}$ preserves composition.

\begin{figure}[t]
  \begin{center}
    \begin{minipage}[m]{0.3\hsize}
      \centering
    \begin{tikzpicture}[
      innode/.style={draw, rectangle, minimum size=\nodesize},
      label/.style={inner sep=0},
      interface/.style={inner sep=0}
      ]
      \node[innode] (ina) {$a$};
      \node[innode] (inb) at ($(ina)+(\hlength,0)$) {$b$};
      \node[interface,anchor=east] (lt) at ($(ina)+(-\hlength/2,\vlength)$) {$\leftpos{1}$}; 
      \node[interface,anchor=east] (lb) at ($(ina)+(-\hlength/2,-\vlength)$) {$\leftpos{2}$}; 
      \node[interface,anchor=west] (r)  at ($(inb)+(\hlength/2,0)$) {$\rightpos{1}$}; 
      \node[label,anchor=south] (inalab) at (ina.north) {$\eve, 1$};
      \node[label,anchor=south] (inblab) at (inb.north) {$\adam, 2$};
      \path[->] (lt)  edge (ina)
                (ina) edge (lb)
                      edge[bend right] (inb)
                (inb) edge[bend right] (ina)
                      edge (r);
    \end{tikzpicture}
    \subcaption{An example of an open parity game.}
    \label{subfig:exOfOpg}
  \end{minipage}
  \begin{minipage}[m]{0.65\hsize}
    \centering
  \scalebox{0.9}{
    $\left(
	  \begin{tikzpicture}[
      baseline=(base.base),
      innode/.style={draw, rectangle, minimum size=\nodesize},
      label/.style={inner sep=0},
      interface/.style={inner sep=0}
      ]
      \node[innode] (ina) {$a$}; 
      \node[interface,anchor=east] (tl) at ($(ina)+(-\hlength/2,\vlength)$) {$\leftpos{1}$}; 
      \node[interface,anchor=east] (ml) at ($(ina)+(-\hlength/2,0)$) {$\leftpos{2}$}; 
      \node[interface,anchor=east] (bl) at ($(ina)+(-\hlength/2,-\vlength)$) {$\leftpos{3}$}; 
      \node[interface,anchor=west] (r)  at ($(ina)+(\hlength/2,0)$) {$\rightpos{1}$}; 
      \node[label,anchor=south] (inalab) at (ina.north) {$\adam, 3$};
      \node[interface] (base) at ($(inalab.north)!0.5!(bl.south)$) {$\phantom{\bullet}$};
      \path[->] (tl)  edge (ina)
                (ml)  edge (ina)
                (bl)  edge (ina)
                (ina) edge (r);
    \end{tikzpicture}  
    \right)$
    ;
    $\left(
    \begin{tikzpicture}[
      baseline=(base.base),
      innode/.style={draw, rectangle,minimum size=\nodesize},
      label/.style={inner sep=0},
      interface/.style={inner sep=0}
      ]
      \node[innode] (inb) {$b$}; 
      \node[interface,anchor=east] (l)  at ($(inb)+(-\hlength/2,0)$) {$\leftpos{1}$}; 
      \node[interface,anchor=west] (tr) at ($(inb)+(\hlength/2,\vlength)$) {$\rightpos{1}$}; 
      \node[interface,anchor=west] (br) at ($(inb)+(\hlength/2,-\vlength)$) {$\rightpos{2}$}; 
      \node[label,anchor=south] (inblab) at (inb.north) {$\eve, 2$};
      \node[interface] (base) at ($(inblab.north)!0.5!(br.south)$) {$\phantom{\bullet}$};
      \path[->] (l)   edge (inb)
                (inb) edge (tr)
                      edge (br);
    \end{tikzpicture} 
    \right)$
    \quad=\quad
    \begin{tikzpicture}[
      baseline=(base.base),
      innode/.style={draw, rectangle,minimum size=\nodesize},
      label/.style={inner sep=0},
      interface/.style={inner sep=0}
      ]
      \node[innode] (ina) {$a$}; 
      \node[innode] (inb) at ($(ina)+(\hlength,0)$) {$b$}; 
      \node[interface,anchor=east] (tl) at ($(ina)+(-\hlength/2,\vlength)$) {$\leftpos{1}$}; 
      \node[interface,anchor=east] (ml) at ($(ina)+(-\hlength/2,0)$) {$\leftpos{2}$}; 
      \node[interface,anchor=east] (bl) at ($(ina)+(-\hlength/2,-\vlength)$) {$\leftpos{3}$}; 
      \node[interface,anchor=west] (tr) at ($(inb)+(\hlength/2,\vlength)$) {$\rightpos{1}$}; 
      \node[interface,anchor=west] (br) at ($(inb)+(\hlength/2,-\vlength)$) {$\rightpos{2}$}; 
      \node[label,anchor=south] (inalab) at (ina.north) {$\adam, 3$};
      \node[label,anchor=south] (inblab) at (inb.north) {$\eve, 2$};
      \node[interface] (base) at ($(inalab.north)!0.5!(bl.south)$) {$\phantom{\bullet}$};
      \path[->] (tl)  edge (ina)
                (ml)  edge (ina)
                (bl)  edge (ina)
                (ina) edge (inb)
                (inb) edge (tr)
                      edge (br);
    \end{tikzpicture} 
  }
\subcaption{An example of sequential composition.}
  \label{fig:compo}
\end{minipage}
  \end{center}
  \caption{Examples of open parity games.}
  \label{fig:exofOpgAndScomp}
\end{figure}
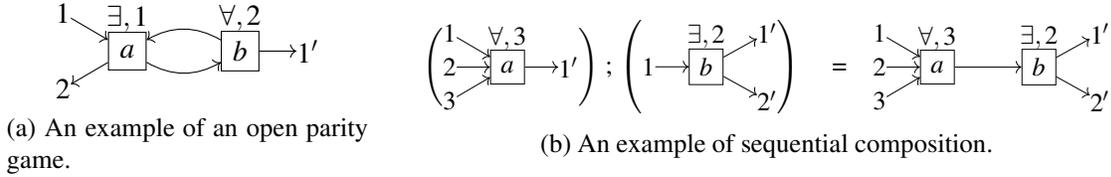

\begin{figure}
  \begin{center}
    \scalebox{1.0}{
    \begin{minipage}[m]{0.29\hsize}
      \centering
      \begin{tikzpicture}[
        innode/.style={draw, rectangle, minimum size=\nodesize},
        label/.style={inner sep=0},
        interface/.style={inner sep=0}
        ]
        \node[innode] (ina) {$a$};
        \node[innode] (inb) at ($(ina)+(\hlength,0)$) {$b$};
        \node[interface] (inc) at ($($(ina)+(0,3*\vlength/2)$)!($(ina)!0.5!(inb)$)!($(ina)+(\hlength,3*\vlength/2)$)$) {};
        \node[interface,anchor=east] (l)  at ($(ina)+(-\hlength/2,0)$) {$\leftpos{1}$}; 
        \node[interface,anchor=east] (bl) at ($(ina)+(-\hlength/2,-\vlength)$) {$\leftpos{2}$}; 
        \node[interface,anchor=west] (r)  at ($(inb)+(\hlength/2,0)$) {$\rightpos{1}$}; 
        \node[label,anchor=south] (inalab) at (ina.north) {$\adam, 3$};
        \node[label,anchor=south] (inblab) at (inb.north) {$\eve, 2$};
        \path[->] (l)   edge (ina)
                  (bl)  edge (ina)
                  (ina) edge (inb)
                  (inb) edge (r);
        \draw ($(inb.north east)!0.33!(inb.south east)$) arc [radius=3*\vlength/4-\nodesize/12,start angle=270,end angle=450];
        \draw[<-] ($(ina.north west)!0.33!(ina.south west)$) arc [radius=3*\vlength/4-\nodesize/12,start angle=270,end angle=90];
        \path ($(inb.north east)!(inc)!(inb.south east)$) edge ($(ina.north west)!(inc)!(ina.south west)$);
      \end{tikzpicture}
      \subcaption{}
      \label{subfig:exTraceCopg}
    \end{minipage}
    \begin{minipage}[m]{0.34\hsize}
      \centering
      \begin{tikzpicture}[
        innode/.style={draw, rectangle, minimum size=\nodesize},
        label/.style={inner sep=0},
        interface/.style={inner sep=0}
        ]
        \node[innode] (ina) {$a$};
        \node[innode] (inb) at ($(ina)+(\hlength,0)$) {$b$};
        \node[interface] (inc) at ($($(ina)+(0,2*\vlength)$)!($(ina)!0.5!(inb)$)!($(ina)+(\hlength,2*\vlength)$)$) {};
        \node[interface] (l)  at ($(ina)+(-\hlength/2,0)$) {};
        \node[interface] (bl) at ($(ina)+(-\hlength/2,-\vlength)$) {};
        \node[interface] (tl) at ($(ina)+(-\hlength/2,\vlength)$) {};
        \node[interface] (r)  at ($(inb)+(\hlength/2,0)$) {};
        \node[interface] (tr) at ($(inb)+(\hlength/2,\vlength)$) {};
        \node[interface] (cl) at ($($(ina.north)!(inc)!(ina.south)$)+(-\hlength/2,0)$) {};
        \node[interface] (cr) at ($($(inb.north)!(inc)!(inb.south)$)+(\hlength/2,0)$) {};
        \node[interface] (ll) at ($(l)+(-3*\hlength/4,0)$) {};
        \node[interface] (lr) at ($(l)+(-\hlength/4,0)$) {};
        \node[interface] (lbl) at ($(l)+(-3*\hlength/4,-\vlength)$) {};
        \node[interface] (lbr) at ($(l)+(-\hlength/4,-\vlength)$) {};
        \node[interface] (rl) at ($(r)+(\hlength/4,0)$) {};
        \node[interface] (rr) at ($(r)+(3*\hlength/4,0)$) {};
        \node[label,anchor=south] (inalab) at (ina.north) {$\adam, 3$};
        \node[label,anchor=south] (inblab) at (inb.north) {$\eve, 2$};
        \path[->] (l)   edge (ina)
                  (tl)  edge (ina)
                  (bl)  edge (ina)
                  (ina) edge (inb)
                  (inb) edge (r)
                        edge (tr)
                  (cr)  edge (cl)
                  (ll)  edge (lr)
                  (lbl) edge (lbr)
                  (rl)  edge (rr);
        \draw[->] ($(rl.south)!(tr)!(rl.north)$) arc [radius=\vlength/2,start angle=270,end angle=450];
        \draw[<-] ($(lr.south)!(tl)!(lr.north)$) arc [radius=\vlength/2,start angle=270,end angle=90];
        \node () at ($(inc.south)!($(inc.south)!0.5!(tl.north)$)!(inc.north)$) {$\parallel$};
        \node (lm) at ($(ll)!0.5!(lr)$) {};
        \node () at ($(lm.south)!($(l)!0.5!(tl)$)!(lm.north)$) {$\parallel$};
        \node (rm) at ($(rl)!0.5!(rr)$) {};
        \node () at ($(rm.south)!($(r)!0.5!(tr)$)!(rm.north)$) {$\parallel$};
        \node (ml) at ($(l)!0.5!(lr)$) {};
        \node () at ($(ml.north)!($(inc.north)!0.5!(bl.south)$)!(ml.south)$) {;};
        \node (mr) at ($(r)!0.5!(rl)$) {};
        \node () at ($(mr.north)!($(inc.north)!0.5!(bl.south)$)!(mr.south)$) {;};
      \end{tikzpicture}
      \subcaption{}
      \label{subfig:traceCopg}
    \end{minipage}
    \begin{minipage}[m]{0.32\hsize}
      \centering
      \begin{tikzpicture}[
        innode/.style={draw, rectangle, minimum size=\nodesize},
        label/.style={inner sep=0},
        interface/.style={inner sep=0}
        ]
        \node[innode] (ina) {$a$};
        \node[innode] (inb) at ($(ina)+(\hlength,0)$) {$b$};
        \node[interface] (inc) at ($($(ina)+(0,2*\vlength)$)!($(ina)!0.5!(inb)$)!($(ina)+(\hlength,2*\vlength)$)$) {};
        \node[interface] (l)  at ($(ina)+(-\hlength/2,0)$) {};
        \node[interface] (bl) at ($(ina)+(-\hlength/2,-\vlength)$) {};
        \node[interface] (tl) at ($(ina)+(-\hlength/2,\vlength)$) {};
        \node[interface] (r)  at ($(inb)+(\hlength/2,0)$) {};
        \node[interface] (tr) at ($(inb)+(\hlength/2,\vlength)$) {};
        \node[interface] (cl) at ($($(ina.north)!(inc)!(ina.south)$)+(-\hlength/2,0)$) {};
        \node[interface] (cr) at ($($(inb.north)!(inc)!(inb.south)$)+(\hlength/2,0)$) {};
        \node[interface] (ll) at ($(l)+(-\hlength/2,0)$) {};
        \node[interface] (lbl) at ($(l)+(-\hlength/2,-\vlength)$) {};
        \node[interface] (rr) at ($(r)+(\hlength/2,0)$) {};
        \node[label,anchor=south] (inalab) at (ina.north) {$\adam, 3$};
        \node[label,anchor=south] (inblab) at (inb.north) {$\eve, 2$};
        \path[->] (ll)  edge (ina)
                  (tl)  edge (ina)
                  (lbl) edge (ina)
                  (ina) edge (inb)
                  (inb) edge (rr)
                        edge (tr);
        \draw ($(rl.south)!(tr)!(rl.north)$) arc [radius=\vlength/2,start angle=270,end angle=450];
        \draw[<-] ($(lr.south)!(tl)!(lr.north)$) arc [radius=\vlength/2,start angle=270,end angle=90];
        \path ($(rl.south)!(inc)!(rl.north)$) edge ($(lr.south)!(inc)!(lr.north)$);
      \end{tikzpicture}
      \subcaption{}
      \label{subfig:traceTsmc}
    \end{minipage}
  }
  \end{center}
  \caption{An example of a cycle and its decomposition using the compact
  closed structure and the traced  monoidal structure.}
\end{figure}
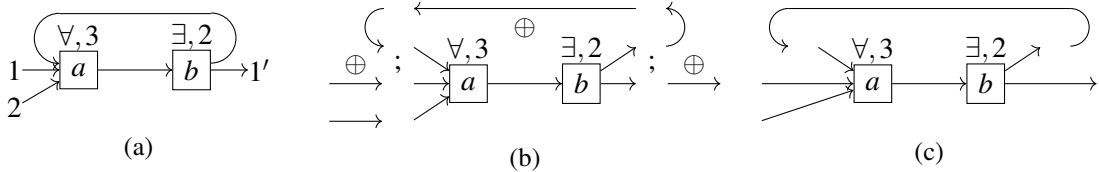

We illustrate our notion of \emph{open parity games} that populates the category $\opg$. Open parity games  are parity games that come additionally with
 interfaces called \emph{open ends},
along which they can be composed.
An example is in Fig.~\ref{subfig:exOfOpg}.
The domain interface consists of two \emph{open ends}, $1$ and
$2$, and the codomain interface simply of $1'$, while the
internal positions are $a$ and $b$, each equipped with a role and a
priority, as usual in parity games.
We give an example of sequential composition in Fig.~\ref{fig:compo}.
There, arrows are composed through the intermediate interfaces \(1', 1\) between the two
games.
We also have a parallel composition \(\parallel\),
and Fig.~\ref{subfig:traceCopg} gives an example of how a cycle (in
Fig.~\ref{subfig:exTraceCopg}) can be defined using sequential and
parallel composition.

\begin{figure}
  \centering
  \begin{tikzpicture}
    \node (sopg) at (0,0) {$\sopg$};
    \node (copg) at (-3,1) {$\copg$};
    \node[anchor=east,inner sep=0] (copg2) at (copg.west) {$\big(\Int(\tcopg)\simeq \big)$};
    \node (opg)  at (3,1) {$\opg$};
    \node[anchor=west,inner sep=0] (opg2) at (opg.east) {$\big(\simeq \Int(\topg)\big)$};
    \path[->] (copg) edge node[above] {\scriptsize $\cto$} (opg)
                     edge node[below left] {\scriptsize $\cts{-}$} (sopg)
              (opg)  edge node[below right] {\scriptsize $\ots$} (sopg);
  \end{tikzpicture}
  \caption{An outline. $\cto$ is the \emph{realization functor} that maps a string diagram to an open parity game; $\ots$ is the \emph{winning position functor} which extends the usual definition of winning positions in parity games; and $\cts{-}$ is the \emph{interpretation functor}.}
  \label{fig:diagram}
\end{figure}

The technical key in Fig.~\ref{fig:diagram} is the identification of
\emph{compact closed} structures.
All the three categories are compact closed;
moreover, we identify the prop $\copg$  to be a \emph{free} compact
closed category in a suitable sense. The functors $\cto$ and $\cts{-}$ arise by the freeness; the commutativity is proved by the freeness, too.

In this paper, we find a new application of props as graphical languages in parity games. It allows one to solve parity games in a compositional manner (Ex.~\ref{ex:extExample}), thanks also to the identification of the right semantical domain (namely $\sopg$) that retains the right level of information in intermediate results. Such compositional solution has multiple potential applications. Firstly,
the categorical structure we identify has a lot in common with those used for HOMC~\cite{grellois2015finitary,DBLP:conf/csl/TsukadaO14}. Therefore we expect we can streamline known HOMC algorithms and reveal their categorical essences. 
Secondly, we will pursue algorithmic applications, such as efficient divide-and-conquer algorithms and those which accommodate blackbox components as part of a game.

\vspace*{.3em}
\noindent\textbf{Organization.}\;
In \S{}\ref{sec:opg}, we introduce open parity
games.
In \S{}\ref{sec:coloredProp}, we define the graphical language
$\copg$ and the realization functor $\cto$.
In \S{}\ref{sec:SemanticCategory}, we define the semantic category
$\sopg$ and the interpretation functor $\cts{-}$.
In \S{}\ref{sec:DenotationalFunctor}, we
 define the winning
position functor $\ots$ and establish the triangle in Fig.~\ref{fig:diagram}.
We also exhibit an example of compositional solution of a parity game in Ex.~\ref{ex:extExample}.
We conclude in \S{}\ref{sec:conclusion}.

\vspace*{.3em}
\noindent\textbf{Related Work.}\;
We use $\copg$ as the graphical language for open parity games.
The use of monoidal categories as graphical languages dates back
to~\cite{penrose1971applications}.
There have been numerous such languages; see~\cite{selinger2010survey}
for a survey.
Languages that compositionally describe graph-like structures are
of particular interest to us:~\cite{fiore2013algebra} describes the
algebra of directed acyclic graphs but does not consider cyclic
structures;~\cite{baez2020open} describes open Petri nets, and
compositionality is achieved ``externally'' by the use of cospans.

In particular, props have been used extensively as graphical languages.
They define graphical languages as models for some
mathematical structures (%
signal flow diagrams~\cite{bonchi2017},
networks~\cite{baez2017prop},
Petri nets~\cite{bonchi2019},
automata~\cite{piedeleu2020string},
and the ZX-calculus~\cite{DBLP:conf/mfcs/CaretteHP19,carette2020}
respectively) and prove that the graphical language is
equivalent to the category that they are studying.
They can therefore transfer properties of the graphical language (for
example, decidability of equivalence of diagrams) to the category they
are studying.
In our work, however, we use the graphical language for
expressing open parity games compositionally, and we are not necessarily
interested in equivalence between $\copg$ and $\opg$ (see also Rem.~\ref{remark:notFaithful}).

Note that our work uses a 2-colored prop for modeling the two possible
directions for edges in an open parity game,
while~\cite{bonchi2019,bonchi2017,lavore2021} only have a single type
of edges, which are undirected.
In~\cite{piedeleu2020string}, the authors use a colored prop to
model different kinds of edges, and in particular, they use two colors to model directed edges.

Kissinger gives a general construction of
the free traced symmetric monoidal categories $\Ft(\Sigma)$,
which are also props~\cite{kissinger2014abstract}.
This is related to the free compact closed category \(\copg\) in the present paper,
as explained in \S{}\ref{sec:coloredProp}.
Free traced monoidal categories are also given in~\cite{katsumata2010categorical}
in the study of attribute grammars, where
many-to-one signatures are treated while Kissinger's paper treats many-to-many signatures.

Another related work is~\cite{ghani2018}, which introduces the concept
of composing games. 
Their approach is mainly applied to economic models, and they use a symmetric monoidal category for compositional game theories. 
However, their framework is different from ours in the sense that the
objects along which games are composed have different meanings: in our
framework, games are composed along graph edges, while in theirs,
games are composed along interfaces describing player choices, game
utility, etc.

\vspace*{.3em}
\noindent\textbf{Notation.}\;
We use the following notations. (i) $\nset{m} \defeq \{1,\dots, m\}$.
(ii) We write the unit and the multiplication of a free monoid \(C^{\ast}\) as \(\emptyword\) and \(\cdot\).
(iii) For \(w\in \{\dr, \dl\}^{\ast}\),
\(\drf{w}\) is the number of \(\dr\) in \(w\),
and \(\dlf{w}\) is the number of \(\dl\) in \(w\).
(iv) $\Natm{M} \defeq \{0,\dots, M\}$ where $M\in \Nat$.
(v) $\Nat_{\geq 1} \defeq \{ i\in \Nat\mid i \geq 1 \}$.
(vi) ${\dual{\dr}} \defeq {\dl}$, ${\dual{\dl}} \defeq {\dr}$, and $\dual{(w_1\cdots w_n)} \defeq \dual{w_n}\cdots \dual{w_1}$ where $w_1\cdots w_n\in \{\dr, \dl\}^{\ast}$.
(vii) $\initf{X}$ is the unique function from $\emptyset$ to \(X\).
(viii) For \(f: A \to C\) and \(g: B \to C\), we write \([f,g]: A+B \to C\)
for the copairing function, i.e., \([f,g](a)=f(a)\) and \([f,g](b)=g(b)\).
(ix) We often omit the injections \(\inj_i : A_i \to A_1 + A_2\), i.e., for example, we may write \(a \in A+B\) for
\(\inj_1(a) \in A+B\), if no confusion happens.


\section{Categories of Open Parity Games}\label{sec:opg}

We introduce the notion of open parity games.
It extends parity games by adding \emph{open ends} to the game,
which are used to define composition of parity games; specifically, we
obtain a compact closed category $\opg$ of open parity games
(Def.~\ref{def:OPG}).

In this paper, we often encounter situations where the structure of
interest can be organized both as a \emph{traced symmetric monoidal
category (TSMC)} or as a \emph{compact closed category (CpCC)}.
(Specifically, we have three such classes of structures, yielding three TSMCs and CpCCs.
See Fig.~\ref{fig:diagram}.)
While our applicational interests lie in the CpCC structures, we work
mostly with the TSMC structures for technical convenience, and use the
Int construction~\cite{joyal1996} to define the CpCC structures from
them ($\opg$ is defined as $\Int(\topg)$, and we show that $\copg$ is
equivalent to $\Int(\tcopg)$ for some signature $\intsig{\ccsp}$).

\subsection{Open Parity Games}

Recall that a \emph{parity game} is a tuple $\mathcal{A} = (Q, E, \rho, \omega)$ where
$(Q, E)$ is a finite directed graph of \emph{positions} and
\emph{edges}, $\rho \colon Q \to \set{\eve,\adam}$ is the \emph{role
function}, and $\omega : Q \rightarrow \Nat$ is the \emph{priority
function}.
An \emph{infinite play} on $\mathcal{A}$ is an infinite sequence $q_0q_1\dots\in Q^{\Nat}$ such that, for all $i\geq 0$, $(q_i, q_{i+1})\in E$. 
A \emph{finite play} is defined similarly.
Let $\Pe$ and $\Pa$ be the sets of finite plays $q_0 \ldots q_n$ such
that $\rho(q_n) = \eve$ or $\rho(q_n) = \adam$, respectively.
An infinite play $q_0q_1\dots$ is \emph{winning} for $\eve$ if the
maximum priority appearing infinitely often in $\omega(q_0)
\omega(q_1) \ldots$ is even.
A finite play $q_0 \dots q_n$ is winning for $\eve$ if $\rho(q_n) =
\adam$.
A \emph{strategy} of $\eve$ is a partial function $\sigma_{\eve}:\Pe
\rightharpoonup Q$
such that $(q_n,\sigma_{\eve}(q_0\dots q_n))\in E$ if $\sigma_{\eve}(q_0\dots q_n)$ is defined.
For any position $q$ and strategies $\sigma_{\eve}$ and
$\sigma_{\adam}$, we denote by $\mathrm{play}^{\sigma_{\eve},
\sigma_{\adam}}_q $ the unique play starting from $q$ and consistent
with both $\sigma_{\eve}$ and $\sigma_{\adam}$.
A strategy $\sigma_{\eve}$ is \emph{winning} for $\eve$ from $q\in Q$
if for all strategies $\sigma_{\adam}$, $\mathrm{play}^{\sigma_{\eve},
\sigma_{\adam}}_q $ is winning for $\eve$.
A position $q\in Q$ is \emph{winning} for $\eve$ if there is a
strategy $\sigma_{\eve}$ winning for $\eve$ from $q$.
We define \emph{open} parity games by extending parity games with open ends.
\begin{definition}[open parity game]
\label{def:openParityGame}
An \emph{open parity game} from \(\bfn{m}\) to \(\bfn{n}\) is a tuple
$(\bfn{m},\bfn{n},Q,E,\rho,M,\omega)$ such that the following conditions are satisfied:
\begin{enumerate}
    \item $\bfn{m} = (\nfr{m},\ \nfl{m})$ and $\bfn{n} = (\nfr{n},\ \nfl{n})$ are pairs of natural numbers,
      where $\bfn{m}$ represents the domain interface of the
      game and $\bfn{n}$ the codomain interface. 
    \item $Q$ is a finite set, whose elements are called \emph{internal positions}.
    \item \label{item:OPGedge}
\(E\) is a relation
    \clovis{give the relation's domain a name}
	  $E\subseteq (\nset{\nfr{m}+\nfl{n}}+Q)\times (\nset{\nfr{n}+\nfl{m}}+Q) $, whose element is called an \emph{edge}.
	  Moreover, 
	  for any $s \in \nset{\nfr{m}+\nfl{n}}$, there is a unique $s' \in \nset{\nfr{n}+\nfl{m}}+Q$ such that $(s,s')\in E$;
	  and similarly 
	  for any $t \in \nset{\nfr{n}+\nfl{m}}$, there is a unique $t' \in \nset{\nfr{m}+\nfl{n}}+Q$ such that $(t',t)\in E$.
    \item $\rho$ is a function $\rho:Q\rightarrow\{\eve, \adam\}$, which assigns a \emph{role} to each internal position.
    \item \(M \in \Nat\) is called the \emph{maximal rank} and $\omega:
      Q\rightarrow\Nat_{M}$ is called the \emph{priority function}.
\end{enumerate}
We call an element of \((\nset{\nfr{m}+\nfl{n}} + \nset{\nfr{n}+\nfl{m}})+Q\) a \emph{position},
one of \(\nset{\nfr{m}+\nfl{n}} + \nset{\nfr{n}+\nfl{m}}\) an \emph{open end},
one of \(\nset{\nfr{m}+\nfl{n}}\) an \emph{entry position}, and
one of \(\nset{\nfr{n}+\nfl{m}}\) an \emph{exit position}.
\end{definition}
We extend the priority function \(\omega\) to 
\(\omega: (\nset{\nfr{m}+\nfl{n}} + \nset{\nfr{n}+\nfl{m}}) + Q\rightarrow\Nat_{M}\)
by \(\omega(i) = 0\) for \(i \in \nset{\nfr{m}+\nfl{n}} + \nset{\nfr{n}+\nfl{m}}\),
i.e., we define the priority of each open end to be 0.

\begin{example}
 The open parity game in Fig.~\ref{subfig:exOfOpg} is the tuple \((\bfn{m}, \bfn{n}, Q, E, \rho, 
 M, \omega)\) where
 \begin{alignat*}{7}
 &\bfn{m} = (1,1), \qquad
 \bfn{n} = (1,0), \qquad
 Q = \{a, b\},
 &\rho(a)&=\exists, \qquad
 &\rho(b)&=\forall,
 \\
 &1 \,E\, a, \qquad
 a \,E\, 2, \qquad
 a \,E\, b, \qquad
 b \,E\, a, \qquad
 b \,E\, 1', \qquad
 &\omega(a)&=1, \qquad
 &\omega(b)&=2, \quad\text{and} \quad
 &&M=2.
 \end{alignat*}
 In Fig.~\ref{subfig:exOfOpg},
 the open end $1$ is the entry position in
 \(\nset{\nfr{m}+\nfl{n}}=\nset{1+0}\).
 The open ends $1'$ and $2$ are the exit positions in \(\nset{\nfr{n}+\nfl{m}}=\nset{1+1}\).
 The two boxes are internal positions in $Q$, with annotations on  roles $\rho$ and priorities $\omega$. 
 As usual, \(E\) is depicted by arrows.
\end{example}

Condition~(\ref{item:OPGedge}) of
Def.~\ref{def:openParityGame} requires that a unique outgoing/incoming edge from/to an entry/exit position, respectively.
This condition can be  enforced by adding some dummy positions. 

The following definition is a first step towards introducing a trace
operator.
\begin{definition}[rightward open parity game]
\label{def:rightDirected}
  An open parity game $\mathcal{A} = (\bfn{m}, \bfn{n}, Q, E, \rho, 
  M, \omega)$ is \emph{rightward} 
  if $\bfn{m} = (\nfr{m}, \nfl{0})$
  and $\bfn{n}= (\nfr{n}, \nfl{0})$
 for some \(\nfr{m}\) and \(\nfr{n}\).
\end{definition}
In the last definition, we require each
open end in \(\bfn{m}\) and \(\bfn{n}\) to be headed in the right.
Note that we do not impose the same requirement on (internal) edges in
$E$---a rightward open parity game may contain cycles.

\subsection{A Traced Symmetric Monoidal Category of Rightward Open Parity Games
}

We shall first define the traced symmetric monoidal category $\topg$ of rightward open parity games.
It yields the compact closed category $\opg$ of open parity games by the Int construction (see Fig.~\ref{fig:diagram}).

In fact, we do so restricting the priorities to be below a certain natural number $M$, talking about $\topg$ and $\opg$. The reason for doing so is discussed in Rem.~\ref{remark:fixMaxRank}.

In what follows, we assume that
a given rightward open parity game \(\mathcal{A}\) is of the form
$\mathcal{A} = \bigl((m^{\mathcal{A}},0),\ (n^{\mathcal{A}},0),\ Q^{\mathcal{A}}, 
\ E^{\mathcal{A}},\ \rho^{\mathcal{A}},\ M,\ \omega^{\mathcal{A}}\bigr)$.
The convention also applies to \(\mathcal{B}\).

We need an equivalence relation on the set of rightward open parity
games to define $\topg$.
For our purpose here, we define the equivalence in terms of
structure-preserving bijections. It is easy to define an equivalence relation on open parity games in the same way.

\begin{definition}[equivalence relation $\sim$ on rightward open parity games]
  \label{def:equivRel}
  We define an equivalence relation $\sim$ on the set of rightward open parity games as follows: 
  $\mathcal{A}\sim\mathcal{B}$ if $m^{\mathcal{A}}= m^{\mathcal{B}}$, $n^{\mathcal{A}}= n^{\mathcal{B}}$, and there is a bijection $\eta:Q^{\mathcal{A}}\rightarrow Q^{\mathcal{B}}$ such that the following conditions are satisfied:
  (i) for $(s,t) \in (\nset{m^{\mathcal{A}}}+Q^{\mathcal{A}})\times (\nset{n^{\mathcal{A}}}+Q^{\mathcal{A}})$, $(s,t)\in E^{\mathcal{A}} \Longleftrightarrow (\bar{\eta}(s),\bar{\eta}(t)) \in E^{\mathcal{B}}$,
(ii) for $s\in Q^{\mathcal{A}}$, $\rho^{\mathcal{A}}(s) = \rho^{\mathcal{B}}(\eta(s))$, and 
(iii) for $s\in Q^{\mathcal{A}}$, $\omega^{\mathcal{A}}(s) = \omega^{\mathcal{B}}(\eta(s))$.
Here we extend \(\eta\) to \(\bar{\eta}: (\Nat + Q^{\mathcal{A}}) \to (\Nat + Q^{\mathcal{B}})\)
by \(\bar{\eta}(n)=n\) for \(n \in \Nat\).
\end{definition}

We define the category $\topg$ as follows.
Objects are natural numbers, and 
a morphism from \(m\) to \(n\) is an equivalence class \(\coset{\mathcal{A}}\) of rightward open
parity games from \((m,0)\) to \((n,0)\).
The identity and composition of morphisms are given by
\(\id_{n} \defeq \coset{\ID_{n}}\) and
\(\coset{\mathcal{A}};\coset{\mathcal{B}} \defeq \coset{\mathcal{A} ; \mathcal{B}}\),
where \(\ID_n\) and \(\mathcal{A} ; \mathcal{B}\) are given in Def.~\ref{def:identityGame} and Def.~\ref{def:seqCompGame} below, respectively.

\begin{definition}[identity]
  \label{def:identityGame}
  For \(n \in \Nat\), we define the \emph{identity game} $\ID_{n}$
  as \(\bigl((n, 0),\ (n, 0),\ \emptyset,\ E,\ \initf{\{\eve, \adam\}},\ M,\ \initf{\Nat_M} \bigr)\)
  where $E = \{(a, a)\mid a\in [n]\}$.
\end{definition}

\begin{wrapfigure}[6]{r}{0pt}
  \begin{minipage}[b]{5.5em}
    \vspace{-15mm}
    \centering
    \begin{tikzpicture}[
            innode/.style={draw, rectangle, minimum size=0.5cm},
            interface/.style={inner sep=0}
            ]
            \node[interface] (rdo1) at (-2cm, 0cm) {$\leftpos{1}$}; 
            \node[interface] (rdo2) at (-2cm, -0.5cm) {$\leftpos{2}$}; 
            \node[interface] (rdo3) at (-2cm, -1cm) {$\leftpos{3}$}; 
            \node[interface] (rcd1) at (-1cm, 0cm) {$\rightpos{1}$}; 
            \node[interface] (rcd2) at (-1cm, -0.5cm) {$\rightpos{2}$}; 
            \node[interface] (rcd3) at (-1cm, -1cm) {$\rightpos{3}$}; 
            \draw[->] (rdo1) to (rcd1);
            \draw[->] (rdo2) to (rcd2);
            \draw[->] (rdo3) to (rcd3);
          \end{tikzpicture} 
      \caption{$\id_3$.\\ \ }
    \label{fig:exIdentity}
  \end{minipage}
  \;
  \begin{minipage}[b]{5.8em}
    \centering
    \vspace{-15mm}
    \begin{tikzpicture}[
      innode/.style={draw, rectangle, minimum size=0.5cm},
      interface/.style={inner sep=0}
      ]
      \node[interface] (rdo1) at (-2cm, 0cm) {$\leftpos{1}$}; 
    \node[interface] (rdo2) at (-2cm, -0.5cm) {$\leftpos{2}$}; 
    \node[interface] (rdo3) at (-2cm, -1cm) {$\leftpos{3}$}; 
    \node[interface] (rcd1) at (-1cm, 0cm) {$\rightpos{1}$}; 
    \node[interface] (rcd2) at (-1cm, -0.5cm) {$\rightpos{2}$}; 
    \node[interface] (rcd3) at (-1cm, -1cm) {$\rightpos{3}$}; 
    \draw[->] (-1.8cm, 0cm) to (-1.2cm, -0.5cm);
    \draw[->] (-1.8cm, -0.5cm) to (-1.2cm, -1cm);
    \draw[->] (-1.8cm, -1cm) to (-1.2cm, 0cm);
  \end{tikzpicture} 
  \caption{$\sigma_{2, 1}$.\\ \ }
  \label{fig:exSwap}
\end{minipage}
\;
\begin{minipage}[b]{12em}
  \centering
  \vspace{-6mm}
  \scalebox{0.7}{
    \begin{tikzpicture}[
      innode/.style={draw, rectangle, minimum size=\nodesize},
      label/.style={inner sep=0},
      interface/.style={inner sep=0}
      ]
        \node[innode] (tina) {$a$}; 
        \node[innode] (tinb) at ($(tina)+(\hlength,0)$) {$b$}; 
        \node[interface,anchor=east] (ttl) at ($(tina)+(-\hlength/2,\vlength)$) {$\leftpos{1}$}; 
        \node[interface,anchor=east] (tml) at ($(tina)+(-\hlength/2,0)$) {$\leftpos{2}$}; 
        \node[interface,anchor=east] (tbl) at ($(tina)+(-\hlength/2,-\vlength)$) {$\leftpos{3}$}; 
        \node[interface,anchor=west] (ttr) at ($(tinb)+(\hlength/2,\vlength)$) {$\rightpos{1}$}; 
        \node[interface,anchor=west] (tbr) at ($(tinb)+(\hlength/2,-\vlength)$) {$\rightpos{2}$}; 
        \node[label,anchor=south] (tinalab) at (tina.north) {$\adam, 3$};
        \node[label,anchor=south] (tinblab) at (tinb.north) {$\eve, 2$};
        \node[interface] (base) at ($(tinalab.north)!0.5!(tbl.south)$) {$\phantom{\bullet}$};
        \path[->] (ttl)  edge (tina)
                  (tml)  edge (tina)
                  (tbl)  edge (tina)
                  (tina) edge (tinb)
                  (tinb) edge (ttr)
                         edge (tbr);
      \node[innode] (bina) at ($(tina)+(0,-5*\vlength/2)$) {$a'$};
      \node[innode] (binb) at ($(bina)+(\hlength,0)$) {$b'$};
      \node[interface] (binc) at ($($(bina)+(0,3*\vlength/2)$)!($(bina)!0.5!(binb)$)!($(bina)+(\hlength,3*\vlength/2)$)$) {};
      \node[interface,anchor=east] (bl)  at ($(bina)+(-\hlength/2,0)$) {$\leftpos{4}$}; 
      \node[interface,anchor=east] (bbl) at ($(bina)+(-\hlength/2,-\vlength)$) {$\leftpos{5}$}; 
      \node[interface,anchor=west] (br)  at ($(binb)+(\hlength/2,0)$) {$\rightpos{3}$}; 
      \node[label,anchor=south] (binalab) at (bina.north) {$\adam, 3$};
      \node[label,anchor=south] (binblab) at (binb.north) {$\eve, 2$};
      \path[->] (bl)   edge (bina)
                (bbl)  edge (bina)
                (bina) edge (binb)
                (binb) edge (br);
      \draw ($(binb.north east)!0.33!(binb.south east)$) arc [radius=3*\vlength/4-\nodesize/12,start angle=270,end angle=450];
      \draw[<-] ($(bina.north west)!0.33!(bina.south west)$) arc [radius=3*\vlength/4-\nodesize/12,start angle=270,end angle=90];
      \path ($(binb.north east)!(binc)!(binb.south east)$) edge ($(bina.north west)!(binc)!(bina.south west)$);
    \end{tikzpicture}
  }
  \caption{Parallel composition of Fig.~\ref{fig:compo} \& Fig.~\ref{subfig:exTraceCopg}.}
  \label{fig:exParallelCompositon}
\end{minipage}
\end{wrapfigure}

Fig.~\ref{fig:exIdentity} shows the identity $\id_3$.

Next, we define the sequential composition \(\mathcal{A};\mathcal{B}\) of rightward open parity games.
The intuition is to connect each exit position of \(\g{A}\) with the corresponding entry position of
\(\g{B}\), and then to hide those interface open ends.
Fig.~\ref{fig:compo} in the introduction illustrates this construction.

\begin{definition}[sequential composition]
  \label{def:seqCompGame}
    Let $\mathcal{A}$ and $\mathcal{B}$ be rightward open parity games and $n^{\g{A}} = m^{\g{B}}$.
    We define the sequential composition $\mathcal{A};\mathcal{B}$ as follows:
  \begin{math}
   \mathcal{A};\mathcal{B} =
  \bigl((m^{\g{A}}, 0),\  (n^{\g{B}}, 0),\ Q^{\mathcal{A}}+Q^{\mathcal{B}},\ 
           E^{\mathcal{A};\mathcal{B}},\  [\rho^{\mathcal{A}}, \rho^{\mathcal{B}}],\  M,\ 
        [\omega^{\mathcal{A}}, \omega^{\mathcal{B}}]\bigr)
  \end{math}, where
  \begin{math}
  E^{\mathcal{A};\mathcal{B}} =
  E^{\mathcal{A}}\setminus \big((\nset{m^{\mathcal{A}}}+Q^{\mathcal{A}})\times\nset{n^{\mathcal{A}}}\big)
    \ +\
  E^{\mathcal{B}}\setminus \big(\nset{m^{\mathcal{B}}}\times (\nset{n^{\mathcal{B}}} + Q^{\mathcal{B}})\big) 
  + \bigl\{\,(s,s')\,\mid\, \text{there exists } a\in \nset{n^{\mathcal{A}}}=\nset{m^{\mathcal{B}}}
       \text{ such that } (s,a)\in E^{\mathcal{A}}\text{ and }(a, s')\in E^{\mathcal{B}}\,\bigr\}.
  \end{math}
\end{definition}

We can show associativity and unitality up to structure-preserving
bijection, which entails that $\topg$ is a category by
Def.~\ref{def:equivRel}.

We also define a parallel (or vertical) composition $\parallel$ of
rightward open parity games,
which gives a monoidal product structure of $\topg$
by \(\coset{\g{A}} \parallel \coset{\g{B}} = \coset{\g{A} \parallel \g{B}}\).
Fig.~\ref{fig:exParallelCompositon} gives an example, notice that the
open ends in the second game need to be shifted, for which we
need the following definition:
for \(l \in \Nat\) and \(s \in \nset{m}+Q\),
let \(\lplus{l}{s} \in \nset{l+m}+Q\) be defined by
\begin{math}
 \lplus{l}{s} = 
  l + s
\end{math}
if
\begin{math}
 s\in \nset{m}
\end{math}, and
\begin{math}
 \lplus{l}{s} = 
 s
\end{math}
if
\begin{math}
 s\in Q
\end{math}.

\begin{definition}[parallel composition]
  \label{def:paraComp}
    Let $\mathcal{A}$ and $\mathcal{B}$ be rightward open parity games.
  The parallel composition $\mathcal{A} \parallel \mathcal{B}$ is defined as follows:
  \begin{math}
      \mathcal{A} \parallel \mathcal{B} =
  \big(  (m^{\mathcal{A}} + m^{\mathcal{B}},0),\  (n^{\mathcal{A}} + n^{\mathcal{B}},0),\ 
  Q^{\mathcal{A}}+Q^{\mathcal{B}},\  E^{\mathcal{A}\oplus \mathcal{B}},\ 
  [\rho^{\mathcal{A}} , \rho^{\mathcal{B}}],\  M,\  
  [\omega^{\mathcal{A}} , \omega^{\mathcal{B}}]\big)
  \end{math}, where $E^{\mathcal{A}\parallel \mathcal{B}}$ is given by
  \begin{math}
  E^{\mathcal{A}\parallel \mathcal{B}}
  = E^{\mathcal{A}}
  \ +\ 
  \big\{
  (\lplus{m^{\g{A}}}{s}, \lplus{n^{\g{A}}}{t})
  \,\big|\,
  (s,t) \in E^{\g{B}}
  \big\}.
  \end{math}
    \clovis{last line is suggestion}
  \end{definition}

The following game swaps the order of entry positons and that of exit positions. 
This makes \(\topg\) a \emph{symmetric} monoidal category.
Fig.~\ref{fig:exSwap} shows the swap game $\sigma_{2,  1}$.
\begin{definition}\dt[swap]
For any $m, n\in \Nato$,
  we define the \emph{swap game} $\sigma_{m,  n}$ as follows:
  \begin{math}
    \sigma_{m, n} =
\big((m+n, 0),\, (n+m, 0),\,\emptyset,\, E^{\sigma_{m,n}},\, \initf{\{\eve, \adam\}},\, M,\, \initf{\Nat_M}\big)
  \end{math},
  where $E^{\sigma_{m,n}}=\{(a, n+a)\mid a\in \nset{m}\}
  \cup\{(m+a, a)\mid a\in\nset{n}\}$.
\end{definition}

Cycles are essential to parity games: without them there would not be any infinite play.
To introduce cycles in rightward open parity games,
we use a trace opearator on $\topg$, as illustrated in
Fig.~\ref{subfig:traceTsmc}.

\begin{definition}[trace operator of \(\topg\)]
  \label{def:traceintopg}
Let \(l\), \(m\), and \(n\) be objects in \(\topg\).
We define the trace operator 
$\tropg{l}{m}{n}: \topg(l+m, l+n)\rightarrow \topg(m, n)$ as follows.
Let \(\g{A} \in \topg(l+m, l+n)\), i.e., let \(m^{\g{A}} = l+m\) and \(n^{\g{A}} = l+n\).
Then, \(\tropg{l}{m}{n}(\coset{\g{A}}) \defeq \coset{\tropg{l}{m}{n}(\g{A})}\)
where
\begin{align*}
\begin{aligned}
 \tropg{l}{m}{n}(\g{A})
 \;=\;&
 \bigl((m, 0),\,(n, 0),\, Q^{\mathcal{A}},\, E^{\tropg{l}{m}{n}(\g{A})},\,\rho^{\mathcal{A}},\, M,\, \omega^{\mathcal{A}}\bigr), \text{ where}
 \\
 E^{\tropg{l}{m}{n}(\g{A})}
 \; = \; &
 \big\{ (s,s') \in (\nset{m}{+}Q^{\g{A}}) {\times} (\nset{n}{+}Q^{\g{A}}) \,\big|\,
 \lplus{l}{s} \,E^{\g{A}}\, a_1 \,E^{\g{A}}\, \cdots \,E^{\g{A}}\, a_k \,E^{\g{A}}\,
 \lplus{l}{s'}
 \text{ for some }
 k \in \Nat, (a_i)_i  \in \nset{l}^k
 \big\}.
\end{aligned}
\end{align*}
\end{definition}

Here is the main result of this section.
With the given definitions, the proof is (lengthy but) routine work.
\begin{theorem}[\(\topg\)]
The data $(\topg,\ \parallel,\ \emptyset,\ \sigma,\ \mathrm{tr})$ defined so far 
constitutes a strict traced symmetric monoidal category, where
  $\emptyset$ denotes the obvious empty game.
\qed
\end{theorem}

\subsection{A Compact Closed Category of Open Parity Games}
To obtain the category $\opg$ of open parity games, we use the
\emph{Int construction}~\cite{joyal1996} (see also~\cite{hasegawa2010note} for some correction).
It is a \emph{free} construction from a traced symmetric monoidal
category $\C$ to a compact closed category $\Int(\C)$.
We briefly explain how it is defined here, but see
\iffull
\else
the full version
\fi
or~\cite{joyal1996} for more details.

Let \(\Cptcc\) be the 2-category of (locally small) CpCCs,
compact closed functors, and monoidal natural transformations.
Note that its 2-cells automatically respect compact closed structures
and are monoidal natural \emph{isomorphisms}~\cite[Proposition~7.1]{JOYAL199320}.
Also, let $\TrSMCg$ be the 2-category of (locally small) TSMCs,
traced symmetric strong monoidal functors, and monoidal natural \emph{isomorphisms}.

Then the Int construction is a left biadjoint to the embedding \(\TrSMCg \to \Cptcc\).
Specifically, given a traced symmetric monoidal cateogry 
\((\C, \otimes, \munit, \sigma, \tr)\),
the category \(\Int(\C)\) is defined as follows:
An object of $\Int(\C)$ is a pair \((X_{+}, X_{-})\) of objects of $\C$.
Then $\Int(\C)\big((X_{+}, X_{-}),\, (Y_{+}, Y_{-})\big) \defeq
\C(X_{+}\otimes Y_{-},\, Y_{+}\otimes X_{-} )$,
and \(\id_{(X_+, X_-)}\defeq \id_{X_+\otimes X_-}\).
Notably, for $f\in \Int(\C)\big((X_{+}, X_{-}),\, (Y_{+}, Y_{-})\big)$ and $g\in\Int(\C)\big((Y_{+}, Y_{-}),\, (Z_{+}, Z_{-})\big)$,
the composite of $f$ and $g$ is defined using the trace operator, namely by  $\trace{Y_{-}}{X_{+}\otimes Z_{-}}{Z_{+}\otimes X_{-}}{\C}\bigl((\sigma_{Z_{+}, Y_{-}}\otimes \id_{X_{-}})\circ (g\otimes \id_{X_{-}})
        \circ (\id_{Y_{+}}\otimes \sigma_{X_{-}, Z_{-}})\circ (f\otimes \id_{Z_{-}})\circ (\sigma_{Y_{-}, X_{+}}\otimes \id_{Z_{-}})\bigr)$. See~\cite{joyal1996} for details, including diagrammatic illustration.

\begin{definition}[$\opg$] \label{def:OPG}
  Let $\opg$ be the compact closed category $\Int(\topg)$ of open
  parity games.
\end{definition}
The following proposition is trivial from the definition.
\begin{proposition}
  A morphism of $\opg$ is an \(\sim\)-equivalence class of open parity
  games (in Def.~\ref{def:openParityGame}).
\qed
\end{proposition}

The compact closed structure is the basis of our compositional approach to parity games.
The structure serves as \emph{game constructors},
and the compact closed functor \(\ots : \opg \to \sopg\) (see Fig.~\ref{fig:diagram})
given in \S{}\ref{sec:DenotationalFunctor} computes  compositionally if an entry position in
a composed game wins.

\section{A Graphical Language of Open Parity Games}
\label{sec:coloredProp}

In this section, we introduce the category \(\copg\) as a graphical language for open parity games.
The category is a \emph{prop}, a symmetric monoidal category version
of the notion of Lawvere theory whose use has been actively
pursued recently~\cite{bonchi2017,baez2017prop,carette2020}.
It gives to open parity games introduced in \S{}\ref{sec:opg} a
language of string diagrams generated by certain generators and
equations.
Moreover, we find that the category \(\copg\) is free in two senses:
(i) as the prop induced by a theory $(\opgsig,\opgequ)$ for open parity games; and (ii) as a compact closed category
$\Int(\tcopg)$ (see Fig.~\ref{fig:diagram}).
The second freeness is exploited in the compositional definition of the interpretation functor $\cts{-}$ in Fig.~\ref{fig:diagram}. 

\subsection{The Graphical Language \(\copg\)}

We define \(\copg\) as a colored prop constructed
from a symmetric monoidal theory $(\opgsig,\opgequ)$.
For the detail of this prop construction, the reader can consult, e.g.,~\cite{carette2020}.

\begin{definition}[$C$-prop, morphism, \(\Cprop{C}\)]
Let $C$ be a set (of \emph{colors}).
A \emph{$C$-prop} is a small strict symmetric monoidal category where
the monoid of all the objects is the free monoid $C^{*}$ of $C$.
A \emph{$C$-prop morphism} between \(C\)-props is a strict symmetric monoidal functor
that is the identity on objects.
We write $\Cprop{C}$ for the category of $C$-props and $C$-prop morphisms.
\end{definition}
In this paper, we consider $\{\dr,\dl\}$-props.
The colors $\dr$ and $\dl$ represent ``rightward'' and ``leftward'', respectively.
They intuitively correspond to \(\nfr{m}\) and \(\nfl{m}\) of \(\bfn{m} = (\nfr{m},\ \nfl{m})\) in an open parity game.

We want to define $\copg$ as a \emph{free} $\set{\dr,\dl}$-prop.
A free \(C\)-prop is generated from a $C$-\emph{symmetric monoidal theory}
(\emph{\csmt{C}} for short), i.e., a pair of a \emph{$C$-signature} and a set of \emph{$C$-equations}.
Intuitively, given a \csmt{C}, morphisms of the corresponding free
$C$-prop are terms built freely from the signature (as well as
sequential and parallel composition) and quotiented by the equations.

\begin{definition}[$C$-signature, morphism, $\Csig{C}$]
A \emph{$C$-signature} is a functor $\Sigma: C^{\ast}\times C^{\ast} \rightarrow \Set$ 
where the free monoid $C^{\ast}$ is thought of as a discrete category.
A \emph{$C$-signature morphism} $f:\Sigma\rightarrow \Sigma'$
is a natural transformation from $\Sigma$ to $\Sigma'$.  
We write $\Csig{C}$ for the category of $C$-signatures and $C$-signature morphisms.
\end{definition}
Thus, \(\Csig{C} = \Set^{C^{\ast} \times C^{\ast}}\).
We define an $\{\dr,\dl\}$-signature $\opgsig$, which is used to
define the graphical language $\copg$.
In the signature $\opgsig$, for each domain $w\in \{\dr,
\dl\}^{\ast}$, codomain $u\in \{\dr, \dl\}^{\ast}$, role $r \in
\set{\adam,\eve}$, and priority $p \in \Nat_M$, there is a single
generator $\nd[w,u]{r}{p}$ that represents the type of nodes in open
parity games with these specific domain, codomain, role, and priority.

\begin{definition}[\(\{\dr,\dl\}\)-signature \(\opgsig\)]
  \label{def:sigForParity}
For \(w,u \in \{\dr,\dl\}^{\ast}\), we define \(\opgsig(w,u) \) as follows: Let $N_{w,u} = \setcomp{\nd[w,u]{r}{p}}{r \in \set{\adam,\eve},
  p \in \Nat_M}$. Then $\opgsig(\emptyword,\dr\cdot\dl) =
  N_{\emptyword,\dr\cdot\dl} \cup \set{d_\dr}$,
  $\opgsig(\dl\cdot\dr,\emptyword) = N_{\dl\cdot\dr,\emptyword} \cup
  \set{e_\dr}$, and $\opgsig(w,u) = N_{w,u}$ otherwise.
  \end{definition}
The generators $d_\dr$ and $e_\dr$ intuitively represent a unit and counit over $\dr$, respectively.
We now turn to equations, for which we first need to define terms of a
\csmt{C}.
They are given by the following free construction.
Let $\Usig^C:\Cprop{C}\rightarrow \Csig{C}$ be the obvious forgetful functor.

\begin{theorem}[\cite{Fregier2009,hackney2015category}]
  \label{thm:UsigHasLeftFsig}
  The forgetful functor $\Usig^C$ has a left adjoint $\Fsig^C:\Csig{C}\rightarrow \Cprop{C}$.
  \qed
\end{theorem}

For the unit \(\eta^{\sig} : \Id_{\Csig{C}} \to \Usig^C \circ \Fsig^C\), 
we identify \((\eta^{\sig}_{\Sigma})_{w,u}(f) \in \Usig^C(\Fsig^C (\Sigma))(w,u) = \Fsig^C (\Sigma)(w,u)\) 
with \(f \in \Sigma(w,u)\) for simplicity of presentation.
  
\begin{definition}\dt[\csmt{C}]
   A \emph{$C$-colored symmetric monoidal theory} (\csmt{C} for short)
   is a tuple $(\Sigma, E, l, r)$ where 
      $\Sigma$ and $E$ are $C$-signatures and
      $l, r : E \rightarrow \Usig^C (\Fsig^C(\Sigma))$ are $C$-signature morphisms.
\end{definition}

We often write simply \((\Sigma, E)\) for \((\Sigma, E, l, r)\).
We call \(\Fsig^C(\Sigma)\) the set of \emph{terms} generated
by \(\Sigma\), and $E$ the set of ($C$-)\emph{equations} in $\Sigma$,
where each $e \in E$ represents the equation $l(e) = r(e)$.

\begin{definition}[SMT \((\opgsig,\opgequ)\)]
  \label{def:SMTforParity}
  We complete the definition of the \(\{\dr,\dl\}\)-SMT
  \((\opgsig,\opgequ,l^{\opgsym}_M,r^{\opgsym}_M)\) by giving the equations:
  $(\dcounit_{\dr} \cplus \id_{\dl})\circ
  (\id_{\dl} \cplus \dunit_{\dr}) = \id_{\dl}$ and $(\id_{\dr} \cplus \dcounit_{\dr})\circ (\dunit_{\dr} \cplus
  \id_{\dr}) = \id_{\dr}$.
  \end{definition}

The \csmt{\{\dr, \dl\}} \((\opgsig,\opgequ)\) describes open
parity games, and the equations in Fig.~\ref{subfig:Equations} represent the coherence conditions of compact closed categories.

In general, a \csmt{C} induces a free $C$-prop $\F(\Sigma, E)$, whose arrows give a graphical language.
\begin{definition}[free prop $\F(T)$ \cite{baez2017prop,carette2020}]
  \label{def:freeByCoequalizer}
  Let  \(T=(\Sigma, E, l, r)\) be a \(C\)-SMT.
  We define a \(C\)-prop $\F(T)$ as the coequalizer of 
  \(\tp{l}, \tp{r}: \Fsig^C(E) \to \Fsig^C(\Sigma)\)
  in \(\Cprop{C}\)
  where
  \(\tp{l}\) and \(\tp{r}\)
  are, respectively, the transposition of \(l, r :  E \to \Usig^C (\Fsig^C(\Sigma))\) in \(\Csig{C}\)
  by \(\Fsig^C \dashv \Usig^C\).
  \end{definition}

\begin{definition}[graphical language \(\copg\)]
  We define \(\copg\) by Def.~\ref{def:SMTforParity} and Def.~\ref{def:freeByCoequalizer}.
\end{definition}
By definition, an object of \(\copg\) is an element of \(\{\dr,\dl\}^{\ast}\),
and a morphism is the class of terms generated by sequential and
parallel composition applied to constants in \(\opgsig\), and
quotiented by the congruence produced by the equations in \(\opgequ\)
(again under sequential and parallel composition).
The prop \(\copg\) is illustrated in Fig.~\ref{subfig:identityAndSwap}, Fig.~\ref{subfig:Signature}, and Fig.~\ref{subfig:Equations}.

\begin{figure}
\centering
\begin{minipage}[b]{0.3\hsize}
  \centering
  \scalebox{0.8}{  
    \begin{tikzpicture}
      \draw[->] (-2cm, 0.0cm)--(-1.5cm, 0.0cm) node[at={(-1.75cm,-0.2cm)}] {$\id_{\dr}$};
      \draw[<-] (-1.2cm, 0.0cm)--(-0.7cm, 0.0cm) node[at={(-0.95cm,-0.2cm)}] {$\id_{\dl}$};
      \draw[->] (-0.4cm, 0.15cm)--(0cm, -0.15cm) node[at={(-0.2cm,-0.4cm)}] {$\sigma_{\dr, \dr}$};
      \draw[->] (-0.4cm, -0.15cm)--(0cm, 0.15cm);
      \draw[->] (0.3cm, 0.15cm)--(0.7cm, -0.15cm) node[at={(0.5cm,-0.4cm)}] {$\sigma_{\dr, \dl}$};
      \draw[<-] (0.3cm, -0.15cm)--(0.7cm, 0.15cm);
      \draw[<-] (1cm, 0.15cm)--(1.4cm, -0.15cm) node[at={(1.2cm,-0.4cm)}] {$\sigma_{\dl, \dl}$};
      \draw[<-] (1cm, -0.15cm)--(1.4cm, 0.15cm);
      \draw[<-] (1.7cm, 0.15cm)--(2.1cm, -0.15cm) node[at={(1.9cm,-0.4cm)}] {$\sigma_{\dl, \dr}$};
      \draw[->] (1.7cm, -0.15cm)--(2.1cm, 0.15cm);
    \end{tikzpicture}
  }
  \subcaption{Identites and swaps}
  \label{subfig:identityAndSwap}
  \end{minipage}
\begin{minipage}[b]{0.18\hsize}
\centering
\scalebox{0.8}{  
  \begin{tikzpicture}[ 
    innode/.style={draw, rectangle, minimum size=0.5cm},
    interface/.style={inner sep=0}
  ]
    \node[interface] (do1) at (2.2cm, 0.2cm) {};
    \node[interface] (do2) at (2.2cm, 0cm) {};
    \node[interface] (do3) at (2.2cm, -0.2cm) {};
    \node[innode] (ina) at (2.7cm, 0cm) {}; 
    \node[anchor=south] (inalab) at (ina.north) {$r, p$};
    \node[anchor=south] (inalab) at (2.7, -1cm) {$\nd[w,u]{r}{p}$};
    \node[interface] (cd1) at (3.2cm, 0.2cm) {};
    \node[interface] (cd2) at (3.2cm, 0cm) {};
    \node[interface] (cd3) at (3.2cm, -0.2cm) {};
    \draw[-] (do1) to (2.45cm, 0.2cm);
    \draw[-] (2.95cm, 0.2cm) to (cd1);
    \draw[-] (do2) to (2.45cm, 0cm);
    \draw[-] (2.95cm, 0cm) to (cd2);
    \draw[-] (do3) to (2.45cm, -0.2cm);
    \draw[-] (2.95cm, -0.2cm) to (cd3);
    \draw[-,dotted] (2.325cm, 0.2cm) to (2.325cm, 0cm);
    \draw[-,dotted] (2.325cm, 0cm) to (2.325cm, -0.2cm);
    \draw[-,dotted] (3.075cm, 0.2cm) to (3.075cm, 0cm);
    \draw[-,dotted] (3.075cm, 0cm) to (3.075cm, -0.2cm);
    \draw[<-] (1cm, 0.2cm) arc (90:270:0.25cm) node[at={(0.9cm,-0.7cm)}] {$\dunit_{\dr}$};
    \draw[<-] (1.4cm, 0.2cm) arc (90:-90:0.25cm) node[at={(1.5cm,-0.7cm)}] {$\dcounit_{\dr}$};
  \end{tikzpicture}
}
\subcaption{Signature \(\opgsig\).}
\label{subfig:Signature}
\end{minipage}
\hspace{10pt}
\begin{minipage}[b]{0.3\hsize}
\centering
\scalebox{0.8}{  
  \begin{tikzpicture}[baseline=-0.5ex,
    innode/.style={draw, rectangle, minimum size=0.5cm},
    interface/.style={inner sep=0}]
    \draw[->] (0cm, 0.5cm) to (0.5cm, 0.5cm);
    \draw[<-] (0cm, 0.5cm) arc (90:270:0.25cm);
    \draw[->] (0cm, -0.5cm) arc (270:450:0.25cm);
    \draw[->] (-0.5cm, -0.5cm) to (0cm, -0.5cm);
  \end{tikzpicture} 
$=$
  \begin{tikzpicture}[baseline=-0.65ex,
    innode/.style={draw, rectangle, minimum size=0.5cm},
    interface/.style={inner sep=0}]
    \draw[->] (0cm, 0cm) to (0.5cm, 0cm);
  \end{tikzpicture} 
}
\quad
\scalebox{0.8}{  
  \begin{tikzpicture}[baseline=-0.5ex,
    innode/.style={draw, rectangle, minimum size=0.5cm},
    interface/.style={inner sep=0}]
    \draw[<-] (0cm, -0.5cm) to (0.5cm, -0.5cm);
    \draw[<-] (0cm, 0cm) arc (90:270:0.25cm);
    \draw[->] (0cm, 0cm) arc (270:450:0.25cm);
    \draw[<-] (-0.5cm, 0.5cm) to (0cm, 0.5cm);
  \end{tikzpicture} 
$=$
  \begin{tikzpicture}[baseline=-0.65ex,
    innode/.style={draw, rectangle, minimum size=0.5cm},
    interface/.style={inner sep=0}]
    \draw[<-] (0cm, 0cm) to (0.5cm, 0cm);
  \end{tikzpicture} 
}
\subcaption{Equations \(\opgequ\).}
\label{subfig:Equations}
\end{minipage}
\begin{minipage}[b]{0.17\hsize}
\centering
\scalebox{0.8}{  
  \begin{tikzpicture}[
    innode/.style={draw, rectangle, minimum size=0.5cm},
    interface/.style={inner sep=0}]
    \node[innode] (ina) at (0.5cm, -0.5cm) {}; 
    \node[anchor=south] (inalab) at (ina.north) {$\eve, 2$};
    \draw[->] (0.05cm, -0.7cm) to (0.25cm, -0.7cm);
    \draw[->] (-0.55cm, -0.7cm) to (0.05cm, -0.7cm);
    \draw[->] (0.05cm, -0.3cm) to (0.25cm, -0.3cm); 
    \draw[<-] (0.05cm, -0.3cm) arc (-90:-270:0.3cm);
    \draw[<-] (0.05cm, 0.3cm) to (1.05cm, 0.3cm);
    \draw[<-] (0.9cm, -0.3cm) arc (-270:-450:0.2cm);
    \draw[<-] (0.75cm, -0.3cm) to (0.95cm, -0.3cm);
    \draw[->] (0.75cm, -0.7cm) to (0.95cm, -0.7cm);
  \end{tikzpicture}
}
  \clovis{make it look nicer}
\subcaption{An example.}
\label{subfig:ExampleInCopg}
\end{minipage}
\caption{Illustration of \((\opgsig,\opgequ)\).}
\label{fig:GraphicalRepresentations}
\end{figure}

\begin{example}
  Fig.~\ref{subfig:ExampleInCopg} 
 is the morphism $(\id_{\dl}\parallel\dcounit_{\dr})\circ (\id_{\dl}\parallel \nd[\dr\cdot\dr,\dl\cdot\dr]{\eve}{2})\circ (\dunit_{\dl}\parallel \id_{\dr})$
 from \(\dr\) to \(\dl\)
 in $\copg$.
 \asd{todo: change the figure and explanation here.}
\end{example}

\subsection{Free Compact Closedness of \(\copg\) and the Full Functor \(\cto\)}
\label{subsec:FreeCCandFull}

We show that the graphical language \(\copg\) is 
a free compact closed category, 
so that we can freely define a compact closed functor from \(\copg\)
to any compact closed category $\C$ with additional structure.
This way we obtain the realization functor \(\cto : \copg \to \opg\) (see Fig.~\ref{fig:diagram}); we show that $\cto$ is full, meaning that every open parity game has a presentation in $\copg$.
Due to the space limitation, we put the detailed information that is written here in the full version.
%

We need some definitions for proving that \(\copg\) is a free compact closed category.
%
We call an object of the category \(\Set^{\{\dr,\dl\}^{\ast} \times \{\dr,\dl\}^{\ast}}\)
a \emph{compact closed signature} (\emph{CCS}, for short).
%
Recall that the signature \(\opgsig\) consists of nodes \(\nd[w,u]{r}{p}\) of open parity games
and the unit \(\dunit_{\dr}\) and counit \(\dcounit_{\dr}\) of compact closed structure.
We define a CCS $\ccsp$ that is a signature of open parity games without the compact closed structure
\(\dunit_{\dr}\) or \(\dcounit_{\dr}\).

\begin{definition}[CCS \(\ccsp\)]
\label{def:TrSigforParity}
We define a CCS \(\ccsp\) by
\(\ccsp(w,u) \defeq
\big\{ \nd[w,u]{r}{p} \,\big|\, r\in\{\eve, \adam\} \text{ and } p\in \Nat_{M}\big\}\).
\end{definition}

To state the free compact closedness of \(\copg\), we define a \emph{valuation}, which defines a way to interpret elements of a signature
into a compact closed category.

\begin{definition}[valuation]
  \label{def:valuationObjOnly}
    For a CCS \(\Sigma\) and a compact closed category \(\C\),
    a \emph{valuation} of \(\Sigma\) into \(\C\) is
    a pair $\big(V_{\dr}, (V_{w,u})_{w, u}\big)$ such that (i) \(V_{\dr} \in ob(\C)\)
    and (ii)
    \(V_{w, u} :\Sigma(w, u)\rightarrow \C(V^{\ast}_{w}, V^{\ast}_{u})\)
    for \(w, u\in \{\dr,\dl\}^{\ast}\)
    where \(V^{\ast}_{w}\) (\(w \in \{\dr,\dl\}^{\ast}\)) is defined as follows:
    $V^{\ast}_{d_1 \ldots d_n} \defeq V_{d_1} \otimes \ldots \otimes V_{d_n}$
   where \(V_{\dl} \defeq \dual{V_{\dr}}\).
\end{definition}

\begin{definition}[action on valuations]
  \label{def:actionObjOnly}
  Given a compact closed functor \(F : \C \to \D\) and
  a valuation \(V\) of \(\Sigma\) into \(\C\),
  the \emph{action} \(\CompCpt{F}{V}\) on \(V\) by \(F\) is defined by
  (i) $\CompCpt{F}{V}{}_{\dr} \defeq F(V_{\dr})$ and
  (ii) $\CompCpt{F}{V}{}_{w, u}(f) \defeq (\phi^{F}_{u})^{-1}\circ F(V_{w, u}(f))\circ \phi^{F}_{w}$,
  where \(\phi^{F}_{w}\) (and \(\phi^{F}_{u}\)) are defined as follows:
  for any \(w = \dr^{i_1}\cdots \dr^{i_n} \in \{\dr,\dl\}^{\ast}\) where
  each \(\dr^{i_j}\) is either \(\dr\) or \(\dual{\dr}\),
  the morphism $\phi^{F}_{w}: 
  F(V_{\dr})^{i_1}\otimes\cdots\otimes F(V_{\dr})^{i_n} \rightarrow F(V_{\dr}^{i_1}\otimes \cdots\otimes V_{\dr}^{i_n})$
  is the isomorphism given by the fact that \(F\) respects the compact closed structures.
\end{definition}

Finally, we prove the free compact closedness of \(\copg\) by using the above definitions.

\begin{theorem}[free compact closedness of \(\copg\)]
\label{thm:freeCptCC}
The prop \(\copg\) is a strict compact closed category.
Furthermore, \(\copg\) is a \emph{free} compact closed category,
i.e.,
there exists a valuation \(\eta_{\ccsp}\) of \(\ccsp\) into \(\copg\)
such that,
for any compact closed category \(\C\)
and any valuation \(V\) of \(\ccsp\) into \(\C\),
there exists a unique (up to iso) compact closed functor 
\(F : \copg \to \C\) such that 
\(\CompCpt{F}{\eta_{\ccsp}} = V\).
\qed
\end{theorem}

 By the general result (Thm.~\ref{thm:freeCptCC}),
we define the realization functor $\cto:\copg\rightarrow\opg$ (see Fig.~\ref{fig:diagram}):
\begin{definition}\dt[realization functor \(\cto\)]
We let the \emph{realization functor} $\cto:\copg\rightarrow\opg$
 be the
 functor determined by Thm.~\ref{thm:freeCptCC} with
\(\cto(\dr)=(1,0)\) and
\(\cto(\nd[w,u]{r}{p})
=
\big((\drf{w},\dlf{w}),\, (\drf{u},\dlf{u}),\,\{\uval\},\, E,\, \allowbreak\{\uval \mapsto r\},\,\allowbreak
M,\, \{\uval \mapsto p\} \big)\)
where
\(E = \big\{(a, \uval) \mid a\in  \nset{\drf{w} + \dlf{u}}\big\}
\cup \big\{(\uval, a) \mid a\in \nset{\drf{u} + \dlf{w}}\big\}\).
\end{definition}

The following theorem says that 
every open parity game in $\opg$ can be represented as a graphical one, i.e., a morphism in \(\copg\).
\begin{theorem}[fullness]
\label{thm:fullness}
The functor $\cto:\copg\rightarrow\opg$ is full.
\qed
\end{theorem}

\begin{remark}
  \label{remark:notFaithful}
  The functor \(\cto\) is not faithful (a one-node example can be easily given).
  In  recent works on props such as~\cite{bonchi2017,piedeleu2020string}, the main interest is in the faithfulness of a semantics functor whose codomain is a well-known semantic category (that of linear relations~\cite{bonchi2017}, automata~\cite{piedeleu2020string}, etc.).
  In this case, faithfulness amounts to the \emph{completeness} of equational axioms.
  We do not share this interest: in Fig.~\ref{fig:diagram}, the codomain of  $\cto:\copg\rightarrow\opg$ is \emph{not} a well-known category, and the value of a corresponding complete equational axiomatization is not clear.
  Faithfulness of the interpretation functor $\cts{-}:\copg\rightarrow \sopg$ (introduced in~\S{}\ref{sec:SemanticCategory}) seems more interesting, since it amounts to an equational characterization of the equivalence of parity games in terms of who is winning.
  The problem seems challenging, however, given the complexity of solving parity games, and we leave it as future work.
  We note that, for our purpose of compositional solution of parity games (see e.g.~Ex.~\ref{ex:extExample}), faithfulness of $\cto$ or $\cts{-}$ is not needed. 
\end{remark}

Kissinger gives a construction for free traced symmetric monoidal
categories \(\Ft(\Sigma)\)~\cite{kissinger2014abstract}.
We show an equivalence \(\copg \simeq \Int(\tcopg)\) in \(\Cptcc\) for some \(1\)-signature \(\intsig{\ccsp}\)
\iffull
\else
(see the full version).
\fi

\section{The Semantic Category of Open Parity Games}
\label{sec:SemanticCategory}
In this section, we define a semantic category of open pairty games $\sopg$.
Grellois and Melli{\`e}s restricted \(\scottl\)
to the full subcategory \(\fscottl\)~\cite{grellois2016semantics,grellois2015finitary} of \emph{finite} preordered sets,
in order to introduce a fixpoint operator on \(\fscottl_{!_M}\)
for some suitable comonad \(!_M\) so that it forms a model of higher-order model checking.
We use this fixpoint operator for the Int construction of \(\sopg\).
We then define the \emph{interpretation functor} $\cts{-} \colon \copg
\to \sopg$.
(The reader may look at Ex.~\ref{ex:extExample} for an
example of computation of $\cts{\g{A}}$ on a concrete open parity game
$\g{A}$.)

\begin{definition}[$\fscottl$~\cite{grellois2016semantics,grellois2015finitary}]
  The category \emph{$\fscottl$} has as objects finite preordered
  sets $A = (|A|, \leq_A)$ and as morphisms \(R: A \relto B\)
  downward-closed binary relations \(R\) between \(A^{\op}\) and
  \(B\): i.e., a binary relation $R\subseteq |A|\times |B|$ such that
  if $a'\geq_A a,\ a \,R\, b,\ b\geq_B b'$, then $a' R\, b'$.
  Composition is defined as usual:
  \(a \,(S \circ R)\, c\) iff \(a \,R\, b\) and \(b \,S\, c\) for some \(b \in |B|\).
  The identity on $A$ is $\id_A = \{(a', a)\mid a\leq_A a'\}$.
\end{definition}

Note that $\fscottl$ has a symmetric monoidal structure induced by the
existence of finite cartesian products $1 = (\{\ast\}, = )$ and $(|A|, \leq_A)\times (|B|, \leq_B) = (|A| + |B|,\, {\leq_A} + {\leq_B})$.

In open parity games, we have two roles: $\eve$ and $\adam$.
For player $\eve$, a choice by $\adam$ is not predictable.
This nondeterminism is represented by the finite powerset comonad $\SPfin$.

\begin{definition}[finite powerset comonad] 
  The \emph{finite powerset comonad} $(\SPfin, \epsilon^{\SPfin},
  \delta^{\SPfin})$ on $\fscottl$ is defined by
  \(\SPfin((|A|, \leq_A)) \defeq (\Pfin(|A|), \leq)\), where
  \(X\leq Y\) iff for any \(x\in X\) there exists \(y\in Y\) such that
  \(x\leq_A y\).
  For \(R: A \relto B\),\, $\SPfin(R) \defeq \{(X, Y)\in\Pfin(|A|)
  \times \Pfin(|B|)\mid \forall y\in Y,\allowbreak\ \exists x\in X,
  (x, y)\in R \}$.
  Then $\epsilon^{\SPfin}_A \defeq \{ (X, a) \in \Pfin(|A|) \times |A|
  \mid \exists x\in X, a\leq_A x\}$ and $\delta^{\SPfin}_A \defeq
  \{ (X,\{Y_1,\dots, Y_n\}) \in \Pfin(|A|) \times \Pfin(\Pfin(|A|))
  \mid Y_1\cup \dots \cup Y_n\leq_{\SPfin(A)} X\}$.
\end{definition}

Priorities that are not greater than $M$ are represented by the coloring comonad $\colmm$.
\begin{definition}[coloring comonad \cite{grellois2016semantics,grellois2015finitary}]
  The \emph{coloring comonad} $(\colmm, \epsilon^{\colmm},
  \delta^{\colmm})$ on $\fscottl$ is defined as follows:
  \(\colmm(|A|, \leq_A) = 
 (\Natm{M} \times |A|,\ \leq_{\colmm \!A})\)
  where \((p, a) \leq_{\colmm \!A} (q, b)\) iff \(p = q\) and
  \(a\leq_A b\).
For \(R: A \relto B\), \ $\colmm(R)\defeq 
\big\{\bigl((p, a), (p, b)\bigr) \in (\Natm{M} {\times} |A|) \times (\Natm{M} {\times} |B|)
\ \big|\ p\in \Natm{M}\ \text{and}\ (a, b)\in R\big \}$.
Then $\epsilon^{\colmm}_A \defeq
 \big\{\bigl((0, a), a'\bigr) \in (\Natm{M} {\times} |A|) \times |A|
 \ \big|\ a'\leq_A a\big\}$ and 
  $\delta^{\colmm}_A \defeq
 \big\{\bigl((\max(p, q), a), (p, (q, a'))\bigr)
 \in (\Natm{M} {\times} |A|) \times (\Natm{M} {\times} (\Natm{M} {\times} |A|))
\ \big|\ a'\leq_A a\big\}$.
\end{definition}

\begin{remark}
  \label{remark:fixMaxRank}
  We expect that this comonad can be extended to a \emph{graded comonad}~\cite{fujii2016towards,katsumata2014parametric}
  so that its Kleisli category interprets all open parity games, 
  without fixing the parity bound \(M\).
  However, we do not take this approach because it is reasonably harmless
  to fix the maximal parity, while
  the use of the complex notion of graded comonad makes it hard to see the essential idea.
\end{remark}

Combining the above notions, we define the comonad \(!_{M}\).
\begin{definition}[comonad \(!_{M}\)~\cite{grellois2016semantics,grellois2015finitary}]
We define a distributive law
$\lambda: \SPfin \circ \colmm \Rightarrow \colmm \circ \SPfin$
 on $\fscottl$ by
\(
  \lambda_{(A, \leq_A)} \defeq \big\{\, (X, (p, Y))
  \,\in\, P(\Natm{M} \times |A|) \times (\Natm{M} \times P(|A|))
  \ \big | \ 
  \forall y\in Y, \exists \ a \in A, (p, a)\in X \text{ and } y\leq_A a 
\,\big\}\),
and we define a comonad \(!_M=(!_M,\epsilon^{!_M},\delta^{!_M})\) on $\fscottl$ by:
(i) \(!_M \defeq \SPfin \circ \colmm\),
(ii) \(\epsilon^{!_M} \defeq \epsilon^{\SPfin} \circ (\SPfin \hcomp \epsilon^{\colmm})\), and
(iii) \(\delta^{!_M} \defeq (\SPfin \hcomp \lambda \hcomp \colmm) \circ (\delta^{\SPfin} \hcomp \delta^{\colmm})\)
where \(\hcomp\) is the horizontal composition of natural transformations.
\end{definition}

In general, the Kleisli category of a comonad inherits the cartesian product from the original category,
and so the Kleisli category $\fscottl_{!_M}$ has a cartesian products given by
\((|A| + |B|, \leq_A + \leq_B)\).
(Furthermore, $\fscottl_{!_M}$ is cartesian
closed~\cite{grellois2016semantics,grellois2015finitary}, though we do
not use this fact.)

In order to give a model of higher-order model checking,
Grellois and Melli{\`e}s introduced
a fixpoint operator \(\fix^{\GM}\) on \(\fscottl_{!_M}\)
to deal with infinite
plays~\cite{grellois2016semantics,grellois2015finitary}.
Its definition is based on the notion of \emph{semantic run-tree}. 
From \(\fix^{\GM}\), we get a trace operator on \(\fscottl_{!_M}\),
because having a fixpoint operator is equivalent to having a trace
operator for a cartesian category~\cite{hasegawa1997recursion}.
We then get a trace operator
\(\tr^{\GM}\) on \(\fscottl_{!_M}^{\op}\), since if \(\C\) is a
traced symmetric monoidal category, then so is \(\C^{\op}\) canonically.

Now we give the definition of \(\tr^{\GM}\)
(the definition of \(\fix^{\GM}\) can be found in 
\iffull
\else
the full version).
\fi
In order to do this, we first adapt the notion of semantic run-tree by
Grellois and Melliès (which we also call \emph{semantic run-tree})
through the correspondance above between fixpoint operators and trace
operators.

\begin{definition}[semantic run-tree for \(\tr^{\GM}\)]
  \label{def:srt}
    Let $R\in\tsopg(D+A, D+B)$ and $a\in |A|$;
    then especially, \(R \subseteq P(\Natm{M} \times (|D|+|B|)) \times (|D|+|A|) \).
    A \emph{semantic run-tree} \(\psi\) for $R$ and $a$ (for the trace operator) is
    a possibly infinite ($\Natm{M}\times (|D|+|A|+|B|)$)-labeled tree \(\psi\) that satisfies the following conditions:
      \begin{enumerate}
        \item \label{def:srt:root} The label of the root of \(\psi\) is $(0, a)\in \Natm{M}\times |A|$.
        \item \label{def:srt:internal} Any node of \(\psi\) that is neither a leaf nor the root has its label in $\Natm{M}\times |D|$.
        \item \label{def:srt:children} For any non-leaf node (possibly being the root that is not a leaf) of \(\psi\) with label $(p, x)\in \Natm{M}\times (|D|+|A|)$,
         let
       $X 
         \subseteq \Natm{M}\times (|D|+|A|+|B|)$
    be the set of the labels of all the children of the node. Then $(X,x)\in R$.
  \item \label{def:srt:leaf} For any leaf node (possibly being the root that is a leaf) of \(\psi\) such that its label belongs to $\Natm{M}\times (|D|+|A|)$
    (rather than \(\Natm{M}\times |B|\)) and is \((p, x)\), we have $(\emptyset, x)\in R$.
      \end{enumerate}
    We write $\tcomp{A}{B}{D}{R}{a}$ for the set of semantic run-trees with respect to \(A,B,D,R\) and \(a\).
    For a semantic run-tree $\psi\in \tcomp{A}{B}{D}{R}{a}$, we define
    $\leaves{\psi} \in |\SPfin(\colmm(B))| = P(\Natm{M} \times |B|)$ as the set of elements 
      $(p, b) \in \Natm{M} \times |B|$ such that
    there exists a leaf \(\ell\) of \(\psi\) such that:
    (i) the label of leaf \(\ell\) is $(p', b)$ for some \(p' \in \Natm{M}\) and
    (ii) $p$ is the maximal priority encountered on the path from the leaf \(\ell\) to the root of $\psi$.
\end{definition}
A semantic run-tree is similar to a (usual) run for a parity game,
except that (i) its branching models $\adam$'s choices, and (ii) it is
induced by a suitable semantic construct $R$ instead of a
graph-theoretic notion of game.
In our use of the notion (\S{}\ref{sec:DenotationalFunctor}), $R$ will
be a ``summary'' of an open parity game, which retains the necessary
data to decide who is winning yet is much smaller than the original
open parity game.

\begin{definition}[trace operator \(\tr^{\GM}\)~\cite{grellois2016semantics,grellois2015finitary}]
  For every \(A,B,D \in \tsopg\), we define a trace operator 
  $\trsopg{D}{A}{B}:\tsopg(D\otimes A, D\otimes B)\rightarrow\tsopg(A,B)$ as follows:
  \[
      \trsopg{D}{A}{B}(R) \defeq 
  \{(\leaves{\psi}, a)\mid \psi \in\tcomp{A}{B}{D}{R}{a} \text{ that meets the
     parity condition}\}
  \]
  where a semantic run-tree meets the \emph{parity condition} if for every infinite path $((p_i,
  x_i))_{i\in \Nat}$,
the maximum priority met infinitely along the path is even
  (i.e., \(\max \{q \mid \#\{i \mid p_i = q\} = \infty\}\) is even).
\end{definition}

Thus \(\fscottl_{!_M}^{\op}\) is a traced symmetric monoidal category.
The trace operator above is used in the (sequential) composition of \(\sopg\) given below.

Now we define the semantic category \(\sopg\) for open parity games.
In \S{}\ref{sec:DenotationalFunctor},
we explain how \(\tsopg\) and \(\sopg\) serve as
the semantic categories in the traced and compact closed structures, respectively,
by giving a suitable \emph{winning-position} functor \(\tots\) from \(\topg\) to \(\tsopg\)
and then by inducing \(\ots\) (see Fig.~\ref{fig:diagram}).

\begin{definition}[semantic category $\sopg$]
By applying the Int construction to the traced symmetric monoidal category \(\fscottl_{!_M}^{\op}\),
we obtain the compact closed category $\Int(\tsopg)$.
\end{definition}

\begin{remark}
  We have \(\tsopg \cong \FinPreord_{T}\), where
  \(\FinPreord\) is the category of finite preordered sets and monotonic functions, and
  (a Kleisli morphism of) the monad \(T\) is of the following form:
  \(
  \tsopg\big(A, B\big)
  \ \cong\ 
  \FinPreord
  \Big(A,
  \big(
  \Pfin^{\uparrow}(\SPfin
  (\mathbf{N}_{M} \times \us{B}
  , \ \leq_{\colmm\!B})
  ), \
  \supseteq
  \big)
  \Big)
  \),
  where \(\Pfin^{\uparrow}\) is the upward-closed powerset.
  This description is closed to the double-powerset style semantics for 2-player games, e.g.,
   in~\cite{hasuo2015generic}.
  \end{remark}

  We want to define an interpretation functor $\cts{-} \colon \copg \to
  \sopg$ that reflects the winning condition on open parity games.
  The idea is that, if $((j_k,p_k)_{k \in n},i) \in \cts{G}$, then player $\eve$ can force any play that
  starts from the entry position corresponding to $i$ in $G$ to end in one of
  the exit positions corresponding to the $j_k$'s while encountering a maximum
  priority of $p_k$.
  By Thm.~\ref{thm:freeCptCC}, we obtain this functor as:
\begin{definition}[interpretation functor $\cts{-}$]
  We define the \emph{interpretation functor} $\cts{-}:\copg\rightarrow \sopg$
  to be the compact closed functor whose action on objects is
  generated by \(\cts{\dr} = \big(([1], =),(\emptyset, =)\big) \in \sopg\) and whose action on morphisms is generated by:
  \begin{align*}
  \cts{\nd[a,b]{\eve}{p}}
  &= \begin{cases}
    \emptyset & (\drf{b} + \dlf{a} = 0)\\
    \big\{\,(T, i)  \ \big|\ 
   (j, p) \in T \text{ for some } j\in [\drf{b} + \dlf{a}] \,\big\} & (\drf{b} + \dlf{a} \neq 0)
  \end{cases}
  \\
  \cts{\nd[a,b]{\adam}{p}} &= \begin{cases}
    P(\emptyset) \times \nset{\drf{a}+\dlf{b}} & (\drf{b} + \dlf{a} = 0)\\
    \big\{\,(T, i) \ \big|\ 
   \{(j, p)\mid j\in [\drf{b} + \dlf{a}] \}\subseteq T\,\big\} & (\drf{b} + \dlf{a} \neq 0).
  \end{cases}
  \end{align*}
  Both morphisms above are from
  \(\cts{a} = (\nset{\drf{a}},\nset{\dlf{a}})\)
  to \(\cts{b} = (\nset{\drf{b}},\nset{\dlf{b}})\)
  in \(\sopg\), i.e.,
  \(P(\Natm{M} \times \nset{\drf{b}+\dlf{a}})
  \relto \nset{\drf{a}+\dlf{b}}\) in \(\fscottl\),
  where powerset is ordered by inclusion.
\end{definition}

\section{Strategies and the Winning Position Functor}
\label{sec:DenotationalFunctor}

In \S{}\ref{subsec:strategies},
we define notions of play and strategy for open parity games
(in the traditional style of graph game),
as well as winning, losing, and \emph{pending} strategies and positions.
We use these definitions in
\S{}\ref{subsec:WinningPositionFunctor} to define the
\emph{winning-position} functor $\tots:\topg\rightarrow\tsopg$, which
gives information that allows compositional computation of winning
positions.
We show that the diagram in Figure~\ref{fig:diagram} commutes,
which gives a justification of our compositional approach to parity games.
In this section, we assume that
a given rightward open parity game \(\mathcal{A}\) is of the form
$\mathcal{A} = (m, n, Q, E, \rho, M, \omega)$.

\subsection{Winning Strategies and Winning Positions for Open Parity Games}
\label{subsec:strategies}

Here we give the notions of strategy and play.
We also define the \emph{denotation} of a strategy/position, which is
how they win, lose, or are pending.
These definitions are given only for \emph{rightward} open parity games,
but we can readily extend them to general open parity games,
because any open parity game is a rightward open parity game
by definition: \(\opg((\nfr{m},\nfl{m}),\,(\nfr{n},\nfl{n}))
=\topg(\nfr{m}+\nfl{n},\nfr{n}+\nfl{m})\).

First we define the notion of strategy.
For an open parity game $\g{A} \colon \nset{m} \to \nset{n}$, a family
\((s_i)_{i \in I}\) of positions in $\nset{m}+\nset{n}+Q$ is called a \emph{position sequence} if
\(I = \Nat_{\ge 1}\) or \(I = \set{1,\ldots,k}\) for some \(k \in \Nat_{\ge 1}\),
(in that case, we also write \((s_i)_{i \in I}\) as \(s_1 \cdots s_k\)).
\begin{definition}[$\eve$-strategy and $\adam$-strategy]
  Let $\mathcal{A}$ be a rightward open parity game from \(m\) to \(n\).
  We define \(\Pe(\g{A}) \defeq \big\{\,s_1\cdots s_k \ \big|\ k \ge 1,\, s_i \in Q\ (i \in \nset{k}),\ (s_i,
   s_{i+1})\in E\ (i\in \nset{k-1}),\, \rho(s_k) = \eve \,\big\}\).
   We often just write $\Pe$ if there is no confusion.
  Then, an \emph{$\eve$-strategy} on $\mathcal{A}$ is a partial function
   $\tau : \Pe \rightharpoonup \nset{n}+Q$ 
   where for any \(s_1\cdots s_k \in \Pe\), (i) if $\tau(s_1\cdots s_k) = s$, then $(s_k, s) \in E$, and (ii) if $\tau(s_1\cdots s_k)$ is undefined, then for all $s\in \nset{n}+Q$, $(s_k,s)\notin E$.
   A $\adam$-strategy on an open parity game $\mathcal{A}$ is defined in the same way,
   by replacing the occurrence of \(\eve\) with \(\adam\) in the above definition.
   The sets of
   $\eve$-strategies and $\adam$-strategies on $\mathcal{A}$ are $\streve{\mathcal{A}}$ and $\stradam{\mathcal{A}}$, respectively.
\end{definition}

A pair of an \(\eve\)-strategy and a \(\adam\)-strategy resolves the
non-determinism in a game to induce a unique play:
\begin{definition}\dt[induced play \(\play{a}{\tau_{\eve}}{\tau_{\adam}}\)]
Let $\mathcal{A}$ be a rightward open parity game from \(m\) to \(n\).
The \emph{induced play} \(\play{a}{\tau_{\eve}}{\tau_{\adam}}\) from an entry position $a\in \nset{m}$ by an $\eve$-strategy
 $\tau_{\eve}$ and a $\forall$-strategy $\tau_{\adam}$
is the (necessarily unique) maximal position-sequence $(s_i)_{i\in I}$ (for the prefix order) such that:
(i) \(a \,E\, s_1\),
(ii) for any $i\in I$, if $\rho(s_i) = \eve$ and $\tau_{\eve}(s_1\cdots s_i)$ is defined,
 then $\tau_{\eve}(s_1\cdots s_i) = s_{i+1}$, and similarly
(iii) for any $i\in I$, if $\rho(s_i) = \adam$ and $\tau_{\adam}(s_1\cdots s_i)$ is defined,
 then $\tau_{\adam}(s_1\cdots s_i) = s_{i+1}$.
\end{definition}

The following notion for a play corresponds to the \emph{winning condition} in (traditional) game theory,
where the condition is two-valued, ``win'' or ``lose''.
Below \(\eve\) and \(\adam\) correspond to ``win'' and ``lose'',
but we have other intermediate results \((m, s_{|I|})\) due to the openness, which we call \emph{pending
states}.
In this paper, we call the following many-valued winning/losing/pending condition
\(\deplay{-}{\mathcal{A}}\) on plays simply \emph{winning condition}.
An infinite position sequence $(s_i)_{i\in\Nat}$ satisfies the \emph{parity condition} if 
the maximum of priorities that occur infinitely in the play is even.
We apply the following notion \(\deplay{-}{\mathcal{A}}\) only to induced plays.
\begin{definition}[winning condition \(\deplay{-}{\mathcal{A}}\) on plays]
\label{def:denotationOfPlay}
    Let $\mathcal{A}$ be a rightward open parity game.
The \emph{denotation} $\deplay{(s_i)_{i\in I}}{\mathcal{A}}$ of a position sequence $(s_i)_{i\in I}$ is defined as
    \begin{alignat*}{2}
        & (m, s_{|I|}) &\text{ if }&I \text{ is finite, }m = \max\{\omega(s_i): i\in I\},\text{ and }
s_{|I|} \text{ is an open end,}\\
        & \eve &\text{ if }&(I \text{ is finite and }\rho(s_{|I|}) = \adam) \text{ or }(I \text{ is infinite and }(s_i)_{i\in I}\text{ satisfies the parity condition}),\\
        & \adam &\text{ if }&(I \text{ is finite and }\rho(s_{|I|}) = \eve) \text{ or }(I \text{ is infinite and }(s_i)_{i\in I}\text{ does not satisfy the parity condition}).
    \end{alignat*}
We call the function \(\deplay{-}{\mathcal{A}}\) the \emph{winning condition} of \(\g{A}\).
\end{definition}

Next, we define the denotation of an \(\eve\)-strategy;
note that an \(\eve\)-strategy
is a strategy for the ``player'' while \(\adam\)-strategies are those
for the ``opponent''.
The denotation is ``lose'' if there is a losing play,
and otherwise is the collection of all the pending states;
if the collection is the empty set, then the denotation is ``win''.
\begin{definition}[denotation of positions and \(\eve\)-strategies]
\label{def:denotationOfStrategy}
    Let $\mathcal{A}$ be a rightward open parity game from \(m\) to \(n\).
    The \emph{denotation} $\destr{(a, \tau_{\eve})}$ of an entry position $a \in \nset{m}$ 
and an $\eve$-strategy $\tau_{\eve}$
is defined by
    \begin{align*}
      & \destr{(a, \tau_{\eve})} \defeq
\left\{
\begin{aligned}
         &\lose 
&\text{if there is } \tau_{\adam} \text{ such that }
         \deplay{\play{a}{\tau_{\eve}}{\tau_{\adam}}}{\mathcal{A}} = \adam\text{,}&
\\
         &\big\{\, \deplay{\play{a}{\tau_{\eve}}{\tau_{\adam}}}{\mathcal{A}} \in \Natm{M} \times \nset{n}
        \ \big|\ \tau_{\adam}\in \stradam{\mathcal{A}}\text{ and } 
      \deplay{\play{a}{\tau_{\eve}}{\tau_{\adam}}}{\mathcal{A}}  \neq \eve \,\big\}
\mspace{-150mu}
      & \text{otherwise}&.
\end{aligned}
\right.
    \end{align*}
\end{definition}

\subsection{The Winning Position Functor \(\ots\)}
\label{subsec:WinningPositionFunctor}

Now we give the central notion of this section, the \emph{winning position} functor \(\ots\),
which is a compact closed functor constructed by the $\Int$-construction of a traced symmetric strong monoidal functor $\tots$.
In the definition of $\tots$,
if we fix an entry position \(a\), then
\(\tots(\g{A})\) (or precisely \(\{T \mid (T,a) \in \tots(\g{A})\}\)) is the upward-closed set generated by
the denotations \(\destr{(a, \tau_{\eve})}\) of \(a\) and
all \(\eve\)-strategies \(\tau_{\eve}\) that does not lose from \(a\):
\clovis{better explanation?}
\begin{definition}[the functor \(\tots\)]
  \label{def:solutionFunctor}
      We define a functor $\tots:\topg\rightarrow \tsopg$.
      The mapping on objects is given by \(\tots(m) \defeq (\nset{m}, = )\), and
      for a morphism
       $\mathcal{A} \in \topg(m, n)$,
       \begin{align*}
           &\tots(\mathcal{A})
  \defeq
        \big\{\, (T,a) \in P(\Natm{M} \times \nset{n}) \times \nset{m} \;\big|\;
       \destr{(a, \tau_{\eve})} \neq \lose \text{ and } \destr{(a, \tau_{\eve})} \subseteq T
   \text{ for some \(\eve\)-strategy } \tau_{\eve}
   \;\big\}.
       \end{align*}
  \end{definition}
  The functor \(\tots\) determines whether an entry position wins, 
  but the precise perspective is as follows.
  As mentioned in the introduction,
  in the traditional notion of (non-open) parity games,
  a position is just either winning or losing, two-valued.
  With the new notion of open ends, however,
  we have the intermediate result of pending states.
  The following definition reflects this idea.

\begin{definition}[winning/losing/pending positions]
    \label{def:winLosePendingPositions}
    Let \(\g{A}\) and \(a\) be a rightward open parity game and an entry position, respectively.
    (i) \(a\) is \emph{winning} if \((\emptyset,a) \in \tots(\g{A})\),
    (ii) \(a\) is \emph{losing} if \((T,a) \notin \tots(\g{A})\) for any \(T\), and
    (iii) \(a\) is \emph{pending} otherwise (i.e., if  \((\emptyset,a) \notin \tots(\g{A})\)
    and \((T,a) \in \tots(\g{A})\) for some \(T \neq \emptyset\)).
\end{definition}

There is an obvious transformation that maps
a traditional parity game \(G\) and position \(x\)
into an open parity game \(\g{A}^{G}_{x} : 1 \to 0 \in \topg\),
where \(1 \in \nset{1}\) points to the internal position \(x\).
The notion of winning/losing defined above agrees with the traditional
one in the following sense (n.b.~there is no pending case):
\begin{proposition}
\label{prop:classical-parity-games}
Given a (traditional) parity game $G$ and a position $x$ in $G$,
\(x\) is winning (resp. losing) in \(G\) iff \(x\) is winning (resp. losing) in \(\g{A}^{G}_{x}\).
\qed
\end{proposition}

The main technical result of this section is stated below, and allows
us to define a compact closed functor \(\ots : \opg \to \sopg\).
\begin{theorem}
  \label{thm:winning-position-functor}
The functor $\tots:\topg\rightarrow \tsopg$ is a traced symmetric strict monoidal functor.
\qed
\end{theorem}

\begin{definition}[winning position functor $\ots$]
  We define the \emph{winning position functor} $\ots$ by $\mathrm{Int}(\tots)$.
\end{definition}

Summarizing all the main results in this paper, we obtain the following theorem:
\ichiro{The following sentence is vague. Perhaps it's better not said. the point you want to convey is
extensively explained in Section~\ref{extExample}, anyway}
\begin{theorem}
\label{thm:triangleOfFunctors}
The triangle in Fig.~\ref{fig:diagram} commutes: $\cts{-} \ \simeq \ \ots \circ \cto$.
\qed
\end{theorem}
We remark that we can obtain a similar result to the above in the TSMC setting by the freeness of \(\tcopg\).

Given any open parity game, 
which can be represented also by a morphism in \(\copg\) 
by the fullness of \(\cto\) (Thm.~\ref{thm:fullness}),
the above Thm.~\ref{thm:triangleOfFunctors} says that
we can calculate whether an entry position is winning, losing, or pending, either 
 (i) by calculating strategies (i.e., by \(\ots\)), or \emph{equivalently}
(ii) by induction (i.e., by \(\cts{-}\)) without calculating strategies.
An elaborated example on how we can compute the denotation of an entry position of an open parity game by the
induction \(\cts{-}\) can be found
below.

Finally, note that 
the notion of winning/losing/pending \emph{position} is defined for \(\sopg\),
and hence is defined also for \(\copg\) and \(\opg\), by using \(\cts{-}\) and \(\ots\), respectively.
On the other hand, the notion of winning/losing/pending \emph{strategy} is defined for morphisms of
\(\opg\) (and hence of \(\copg\), too) but not of \(\sopg\).
In particular, we can conclude that we have given an \emph{abstract} (or extensional)
semantics for open parity games, by eliminating the information of strategies.

\begin{figure}
  \newcommand{\lsp}{10pt}
  \begin{center}
    \begin{minipage}[b]{0.17\hsize}
      \centering
      \begin{tikzpicture}[
          innode/.style={draw, rectangle, minimum size=0.5cm},
          interface/.style={inner sep=0}
          ]
          \node[interface] (rdo1) at (-1.5cm, -2cm) {$\leftpos{1}$}; 
          \node[innode] (in3) at (-0.5cm, -1cm) {$a$};
          \node[anchor=south] (in3lab) at (in3.north) {$\eve, 1$};
          \node[innode] (in4) at (1cm, -1cm) {$b$};
          \node[anchor=south] (in4lab) at (in4.north) {$\adam, 1$};
          \node[innode] (in5) at (1cm, -2cm) {$c$};
          \node[anchor=south] (in4lab) at (in5.north) {$\eve, 2$};
          \draw[->] (rdo1) to (in3);
          \draw[->] (in3) to (in4);
          \draw[->] (in5) to (in3);
          \draw[<-] (1.3cm, -2cm) arc [radius=0.5, start angle = 270, end angle=450];
          \draw[->] (-0.2cm, -0.8cm) to (0.2cm, -0.8cm);
          \draw[->] (0.1cm, -0.8cm) arc [radius=0.3, start angle = 270, end angle=450];
          \draw[<-] (-1.2cm, -0.2cm) to (0.2cm, -0.2cm);
          \draw[->] (-1.2cm, -0.2cm) arc [radius=0.3, start angle = 90, end angle=270];
          \draw[->] (-1.2cm, -0.8cm) to (-0.8cm, -0.8cm);
      \end{tikzpicture}
      \subcaption{Parity game $\mathcal{A}$.}
      \label{fig:exOpenParityGame}
    \end{minipage}
    \hspace{\lsp}
    \begin{minipage}[b]{0.25\hsize}
      \centering
      \begin{tikzpicture}[
        innode/.style={draw, rectangle, minimum size=0.5cm},
        interface/.style={inner sep=0}
        ]
        \node[interface] (rdo1) at (-1.5cm, -2cm) {$\leftpos{1}$}; 
        \node[interface] (lcdo1) at (-0.5cm, -0.2cm) {$\rightpos{1}$}; 
        \node[interface] (rcdo2) at (-0.5cm, -0.8cm) {$\rightpos{2}$}; 
        \node[interface] (rcdo3) at (-0.5cm, -2cm) {$\rightpos{3}$}; 
        \draw[->] (rdo1) to (rcdo3);
        \draw[->] (-0.8cm, -0.2cm) arc [radius=0.3, start angle = 90, end angle=270];
      \end{tikzpicture}
      \subcaption{Open parity game $\subopg{\g{A}}_1$.}
      \label{fig:exSubOpenParityGame1}
    \end{minipage}
    \hspace{\lsp}
    \begin{minipage}[b]{0.25\hsize}
      \centering
      \begin{tikzpicture}[
        innode/.style={draw, rectangle, minimum size=0.5cm},
        interface/.style={inner sep=0}
        ]
            \node[interface] (ldo1) at (-1.5cm, -0.2cm) {$\leftpos{1}$}; 
            \node[interface] (rdo2) at (-1.5cm, -0.8cm) {$\leftpos{2}$}; 
            \node[interface] (rdo3) at (-1.5cm, -2cm) {$\leftpos{3}$}; 
            \node[innode] (in3) at (-0.5cm, -1cm) {$a$};
            \node[anchor=south] (in3lab) at (in3.north) {$\eve, 1$};
            \node[interface] (rcdo1) at (1cm, -1cm) {$\rightpos{1}$}; 
            \node[interface] (lcdo2) at (1cm, -2cm) {$\rightpos{2}$}; 
            \draw[->] (rdo3) to (in3);
            \draw[->] (in3) to (rcdo1);
            \draw[->] (lcdo2) to (in3);
            \draw[->] (-0.2cm, -0.8cm) to (0.2cm, -0.8cm);
            \draw[->] (0.1cm, -0.8cm) arc [radius=0.3, start angle = 270, end angle=450];
            \draw[<-] (ldo1) to (0.2cm, -0.2cm);
            \draw[->] (rdo2) to (-0.8cm, -0.8cm);
      \end{tikzpicture}
      \subcaption{Open parity game $\subopg{\g{A}}_2$.}
      \label{fig:exSubOpenParityGame2}
    \end{minipage}
    \hspace{\lsp}
    \begin{minipage}[b]{0.23\hsize}
      \centering
      \begin{tikzpicture}[
        innode/.style={draw, rectangle, minimum size=0.5cm},
        interface/.style={inner sep=0}
        ]
        \node[interface] (rdo1) at (-1.5cm, 0cm) {$\leftpos{1}$}; 
        \node[interface] (ldo2) at (-1.5cm, -1.5cm) {$\leftpos{2}$}; 
        \node[innode] (in3) at (-0.5cm, 0cm) {$b$};
        \node[anchor=south] (in3lab) at (in3.north) {$\adam, 1$};
        \node[innode] (in4) at (-0.5cm, -1.5cm) {$c$};
        \node[anchor=south] (in4lab) at (in4.north) {$\eve, 2$};
        \draw[->] (rdo1) to (in3);
        \draw[->] (in4) to (ldo2);
        \draw[->] (in3) to (0.3cm, 0cm);
        \draw[->] (0.3cm, -1.5cm) to (in4);
        \draw[<-] (0.3cm, -1.5cm) arc [radius=0.75, start angle = 270, end angle=450];
      \end{tikzpicture}
      \subcaption{Open parity game $\subopg{\g{A}}_3$.}
      \label{fig:exSubOpenParityGame3}
    \end{minipage}
  \end{center}
  \caption{An extended example.}
  \label{fig:extendedExample}
\end{figure}
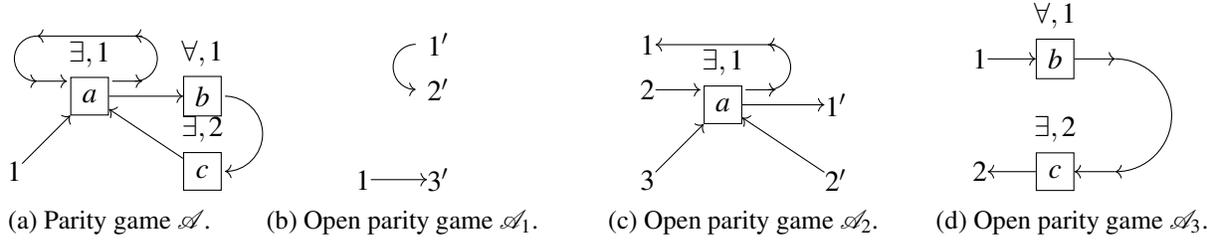

\begin{example}
  \label{ex:extExample}
  Let $\mathcal{A}$ be the open parity game in
  Fig.~\ref{fig:exOpenParityGame}.
  We want to check whether the position $\leftpos{1}$ is winning
  by composing the interpretations of
  $\mathcal{A}$'s subgames.

  Concretely, $\mathcal{A}$ is divided as $\subopg{\mathcal{A}}_{3}
  \circ \subopg{\mathcal{A}}_{2}\circ \subopg{\mathcal{A}}_{1}$ with
  $\subopg{\mathcal{A}}_{1}$, $\subopg{\mathcal{A}}_{2}$, and
  $\subopg{\mathcal{A}}_{3}$ shown in
  Fig.~\ref{fig:exSubOpenParityGame1},~\ref{fig:exSubOpenParityGame2}, 
  and~\ref{fig:exSubOpenParityGame3}, 
  respectively (note that open ends are labelled
  using prop-style ordering, and while the ordering in $\sopg$ is
  different, we keep the same notations for readability).
  It follows directly by unfolding definitions and by compact
  closedness of $\cts{-}$ that
  $\cts{\subopg{\mathcal{A}}_{1}} = \setcomp{(T,\leftpos{1})}{(0,
  \rightpos{3}) \in T} \cup \setcomp{(T, \rightpos{1})}{(0,
  \rightpos{2}) \in T}$,
  $\cts{\subopg{\mathcal{A}}_{2}} = \setcomp{(T,i)}{i \in
  \set{\leftpos{2},\leftpos{3},\rightpos{2}}, \exists j \in
  \set{\leftpos{1},\rightpos{1}}, (1,j) \in T}$,
  and $\cts{\subopg{\mathcal{A}}_{3}} = \setcomp{(T,\leftpos{1})}{(2,
  \leftpos{2}) \in T}$, which are indeed the expected resuts.
  For example, to compute $\cts{\g{A}_1}$, we can decompose $\g{A}_1$
  as $\dunit_\dl \parallel \id_\dr$, so $\cts{\g{A}_1} =
  \dunit_{\cts{\dl}} \parallel \id_{\cts{\dr}}$, which can easily be
  computed from the definition of the identity in $\tsopg$.

  In order to compute the composition of two interpretations in
  $\sopg$, we need to compute a trace, and therefore semantic
  run-trees (we can avoid it in $\subopg{\g{A}}_{2}$ and
  $\subopg{\g{A}}_3$ above because they can be reorganized so that
  composition involves a trivial trace).
  In semantic run-trees corresponding to $\cts{\subopg{\g{A}}_2} \circ
  \cts{\subopg{\g{A}}_1}$, there must be no infinite path (corresponding
  to the intuition that the only infinite path is losing), and the only
  possible leaf is $\rightpos{1}$ (the only exit position), while 
conditions in Def.~\ref{def:srt}
(involving $\cts{\subopg{\g{A}}_2}$ and
  $\cts{\subopg{\g{A}}_1}$ above) ensure that $(1,\rightpos{1})$ must be one
  of the leaves, which gives $\cts{\subopg{\g{A}}_2 \circ
  \subopg{\g{A}}_1} = \setcomp{(T,i)}{i \in \set{\leftpos{1},
  \rightpos{2}}, (1,\rightpos{1}) \in T}$.

  Similarly, to compose $\cts{\subopg{\g{A}}_3}$ with \(\cts{\subopg{\g{A}}_2 \circ
  \subopg{\g{A}}_1}\), we also have to compute the corresponding
  semantic run-trees.
  Here, there can be no leaves (no exit positions), and the run-tree
  corresponding to taking the loop infinitely meets the parity
  condition (because all its nodes are $(2,\rightpos{1})$, except for the
  root), so $(\emptyset,1)$ is in the interpretation, whence
  $\cts{\g{A}} = \setcomp{(T,1)}{\text{true}} = \{(\emptyset,1)\}$. 
  Therefore, $\leftpos{1}$ is a winning poistion in $G$.
\end{example}

\section{Conclusions and Future Work}
\label{sec:conclusion}

We have given a compositional approach to parity games by exhibiting
their underlying compact closed structure.
Parity games can be composed by considering open ends, and we
defined a prop that gives a graphical language to describe such open
parity games.
At the semantic level, we have given a notion of winning/losing
positions that takes open ends into account.
It retains enough information to be compositional, but is still
extensional, as it can be computed without referring to starategies.

\asd{todo: write future work and/or concluding remark}

The current semantic category is a \emph{strategy-insensitive} model,
in that it only keeps track of the (non)existence of a winning strategy,
while \emph{strategy-sensitive} models should keep track of all
strategies (perhaps up to some suitable equivalence).
Strategy-sensitive models can easily be obtained when restricting to
\emph{history-free} strategies.
One future work is to find some history-dependent strategy-sensitive model.

It could be fruitful to deepen the link between the existing body of
work on props and our use of props in this work.
For example, by showing equivalence between the graphical language and
the category of open parity games, we could get decidability results
on parity games from the syntax, as in~\cite{bonchi2017}.
\asd{Analyzing the complexity of the inductive computation of win/lose
and of the fullness proof seems interesting.}

Another possible future work is the coalgebraic treatment of open parity games.
There seems a bijective correspondence between open parity games
up to some notion of bisimilarity and
pairs of functions
\(\nset{m} \to Q\) and
\(Q \to \Pfin(Q + \nset{n}) \times \{\eve,\adam\} \times \Nat_M\)
up to the bisimilarity of coalgebras.
Then there might exist a compact closed structure in the category of coalgebraic open parity
games up to the bisimilarity.
However, it seems subtle to give the same categorical structure in the category of,
say, the above form of coalgebras up to \emph{isomorphism}.
Also, it might be interesting to give a functor (similar to \(\cto\)) from some graphical category to
the category of coalgebras up to bisimilarity by using bialgebraic methods~\cite{klin2011bialgebras,turi1997towards}.
We did not take the bisimilarity approach in this paper
because in game theory we basically consider the level up to isomorphism, say, for complexity.

\clovis{parity games + $\sopg$ from Grellois--Melliès $\leadsto$ maybe
link between our approach and Tsukada--Ong's HOMC paper.}

\bibliographystyle{eptcs}
\bibliography{generic}

\iffull
\else 
\fi
\nocite{*}
\end{document}